\documentclass[11pt]{article}
\usepackage{authblk}

\usepackage{graphicx}

\usepackage{amsmath,amssymb,amsthm}
\newtheorem{lemma}{Lemma}
\newtheorem{proposition}{Proposition}
\newtheorem{theorem}{Theorem}

\usepackage{graphicx,float}
\usepackage{color}
\usepackage[margin = 1in]{geometry}
\numberwithin{equation}{section}
\numberwithin{lemma}{section}
\numberwithin{proposition}{section}
\numberwithin{figure}{section}

\newcommand{\pl}{\partial}

\newcommand{\ubar}{\overline{u}}
\newcommand{\R}{\mathbb{R}}
\newcommand{\C}{\mathbb{C}}
\newcommand{\beq}{\begin{equation}}
\newcommand{\eeq}{\end{equation}}
\newcommand{\half}{\frac12}
\begin{document}

\title{The Kawahara Equation: Traveling Wave Solutions Joining Periodic Waves }

\author[1]{Patrick Sprenger} 
\author[2]{Thomas J. Bridges}
\author[1]{Michael Shearer} 
\affil[1]{Department of Mathematics, North Carolina State University}
\affil[2]{Department of Mathematics, University of Surrey}




\date{}
\maketitle

\begin{abstract}
The Kawahara equation is a weakly nonlinear long-wave model of dispersive waves that emerges when leading order dispersive effects are in balance with the next order correction. Traveling wave solutions of the Kawahara equation  satisfy a fourth-order ordinary differential equation in which the traveling wave speed is a parameter. The fourth order equation has Hamiltonian structure and admits a two-parameter family of single-phase periodic solutions with varying speed and Hamiltonian. A set of jump conditions is derived for pairs of periodic solutions with equal speed and Hamiltonian. These are necessary conditions for the existence of traveling waves that asymptote to the periodic orbits at $\pm \infty$. Bifurcation theory and parameter continuation are used to construct  multiple solution branches of the jump conditions.  For each  pair of compatible periodic solutions, the heteroclinic orbit representing the traveling wave is constructed  from the intersection of stable and unstable manifolds of the periodic orbits. Each branch terminates at an equilibrium-to-periodic solution in which the equilibrium is the background for a solitary wave that connects to the associated periodic solution.

\end{abstract}

\section{Introduction}

The Kawahara equation
\begin{align}\label{eq:kawahara}
u_t + uu_x + \alpha u_{xxx} + u_{xxxxx} = 0,
\end{align}
is a model for weakly nonlinear, dispersive waves for which the third and fifth order dispersive terms are both significant, their balance indicated by the parameter $\alpha $. Such a balance occurs in a variety of physical contexts. For example, shallow water waves are described by the Kawahara equation when surface tension and gravity effects are comparable, corresponding to Bond numbers near $1/3$ \cite{hunter_existence_1988}.  The equation also arises as a continuum model for chains of coupled oscillators \cite{gorshkov_existence_1979} under specific interaction laws, for  magneto-acoustic waves propagating at a critical angle relative to an applied magnetic field \cite{kakutani_weak_1969,kawahara_oscillatory_1972}, and for nonlinear optical systems \cite{baqer_modulation_2020,el_radiating_2016,wai_nonlinear_1986,webb_generalized_2013}.

Traveling wave solutions (TWs) of the Kawahara equation \eqref{eq:kawahara} satisfy a fifth order ODE that can be integrated once, revealing the Hamiltonian structure of the resulting fourth order equation. The modulations of periodic TW solutions of the PDE \eqref{eq:kawahara}  can be studied using Whitham modulation theory \cite{whitham_linear_1974}, in which a periodic TW is modulated to vary on slow time and space scales. The solutions are described to leading order by a system of first order PDEs in conservative form for variables $(\ubar, a, k)$ representing the space-time variations of  the average, amplitude and wavenumber, respectively,  of the underlying periodic solution. Because the modulation equations constitute a system of conservation laws, they can be studied using the well-developed theory of such equations. A next step is to investigate how shock wave solutions of the modulation equations relate to traveling wave solutions of the Kawahara PDE \eqref{eq:kawahara}. A key element of this approach is to recognize that each constant value of the triple $(\ubar, a, k)$ represents a periodic TW with a specific wave speed. Thus, a jump between constant values of $(\ubar, a, k)$ becomes a necessary condition for the existence of a TW solution of \eqref{eq:kawahara} connecting the two associated periodic solutions, provided their wave speeds are the same.  In the paper of Sprenger and Hoefer \cite{sprenger_discontinuous_2020}, the authors made this connection for  the KdV5 equation (the Kawahara equation \eqref{eq:kawahara} with $\alpha = 0$) and constructed multiple traveling wave solutions of the PDE using numerical computations. 

In this manuscript, a somewhat different approach is taken. We  work directly with the fifth order ODE to derive jump conditions for traveling waves approaching distinct periodic waves in the far-field. We use tools from dynamical systems to construct traveling wave solutions, specifically the characterization of stable and unstable manifolds of periodic solutions of the fourth-order ODE, in which the Hamiltonian structure is of crucial importance. 

\subsection{Preliminaries}
 
\subsubsection{Dispersion relation}
The Kawahara equation, linearized about a background $u=\ubar,$ is the constant coefficient PDE
\begin{equation} 
\label{lin_PDE}
v_t+\ubar v_x+\alpha v_{xxx} +v_{xxxxx}=0.
\end{equation}
It has the linear dispersion relation 
\begin{equation}\label{eq:disp1}
\omega_0=\omega_0(k,\ubar;\alpha)=\ubar k - \alpha k^3 + k^5,
\end{equation}
in which $k$ is the wavenumber of solutions proportional to $v(x,t)=e^{i(kx-\omega_0 t)}.$ Consequently, the group and phase velocities are given respectively by 
\begin{align}
\label{eq:group1}
c_{\rm g}(k,\ubar;\alpha) &= \pl_k \omega_0=\ubar-3\alpha k^2+5k^4,\\ \label{eq:phase1}
c_{\rm p}(k,\ubar;\alpha) &= \frac {\omega_0}{k}=\ubar-\alpha k^2+k^4, 
\end{align}
and are plotted in Figure \ref{fig:vel_plot} for $\alpha \in \{\pm 1,0\}$

\begin{figure}[H]
\begin{center}
  \includegraphics[scale=0.55]{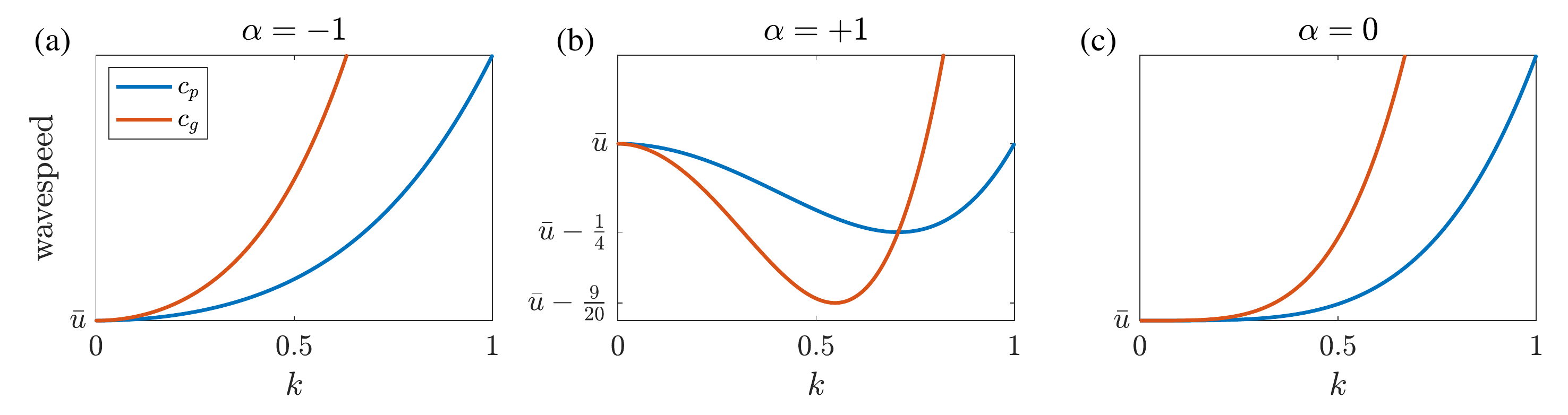}
\caption{Graphs of the phase velocity $c_p$ (blue) and group velocity $c_g$ (red) for (a) $\alpha = -1$ (b) $\alpha = +1$ and (c) $\alpha = 0$.}
\label{fig:vel_plot}
\end{center}
\end{figure}

The parameter $\alpha$ plays a significant role in the properties of the linear, dispersive waves and the nonlinear solutions in what follow. We  distinguish between the three cases $\alpha = 0$, $\alpha > 0$ and $\alpha < 0$. For fixed background mean $\ubar,$ the dispersion relation is a convex function of $k\geq 0$ if $\alpha\leq 0.$ However, for $\alpha>0,$ the dispersion relation is not convex and there are nonzero values of $k$ at which the phase and group velocities each attain a minimum. The minimum values are labeled in Figure~\ref{fig:vel_plot}(b). Moreover, there is a range of phase velocities corresponding to pairs of wavenumbers that satisfy the resonance condition 
 \begin{align}\label{eq:resonance}
     c_{\rm p}(k_1,\ubar)=c_{\rm p}(k_2,\ubar),
 \end{align} 
 with $0<k_1<k_2.$  Specifically, wavenumbers $k_1$ and $k_2$ correspond to the same phase velocity if $k_1^2+k_2^2=\alpha.$ Moreover, the corresponding linear waves are harmonically related if $k_2=nk_1$ for  integer $n\geq 2,$ with $k_1^2=\alpha/(1+n^2).$ 
 
\subsubsection{Scaling Properties} \label{scaling1}
 When $\alpha = 0$, equation Eq. \eqref{eq:kawahara} becomes the fifth order Korteweg-de Vries equation (KdV5), which is invariant under the scaling of variables 
\begin{equation}
 \label{scale1}
 u = b \tilde{u},\quad  t = b^{-5/4} \tilde{t}, \quad x = b^{-1/4} \tilde{x},
\end{equation}
for any $b>0.$  For fixed $\alpha\neq 0,$ the Kawahara equation is not invariant under this scaling. However, scaling the variables as in \eqref{scale1} corresponds to a scaling of $\alpha,$ so that the two cases $\alpha=\pm 1$ represent all non-zero values of $\alpha.$ To see this, we insert the scaled variables into the Kawahara equation \eqref{eq:kawahara},
  \begin{align}
 \tilde{u}_{\tilde{t}} + \tilde{u}\tilde{u}_{\tilde{x}} + \alpha b^{-1/2} \tilde{u}_{\tilde{x}\tilde{x}\tilde{x}} + \tilde{u}_{\tilde{x}\tilde{x}\tilde{x}\tilde{x}\tilde{x}} = 0. 
    \end{align}
  Setting $b = \alpha^2$ we observe the coefficient of the third order dispersive term is $\alpha b^{-1/2}=\pm 1.$ We then recover the Kawahara equation \eqref{eq:kawahara} with $\alpha = \pm 1.$  
  
Equation \eqref{eq:kawahara} is also invariant with respect to the Galilean transformation 
 \begin{align}\label{eq:galilean}
 u(x,t) \to u(x-\ubar t,t) + \ubar.
 \end{align}
Thus, the effect of adding the constant background $\ubar$ is to change the speed of the spatial frame by the same constant.

\subsubsection{Traveling Waves and Hamiltonian Structure}  \label{TW_Hamiltonian}
    Traveling wave solutions of the Kawahara equation \eqref{eq:kawahara} have the form
\begin{equation}\label{tw1}
u(x,t) = f(\xi), \quad \xi = x - ct, 
\end{equation}
where $f$ is a smooth real function defining the profile of the wave and  the speed $c$ is also known as the phase velocity of the TW.  Substituting \eqref{tw1} into \eqref{eq:kawahara}, we obtain the fifth order ordinary differential equation 
\begin{align}\label{eq:profile_eq}
-c f' + ff' + \alpha f'''  + f^{(5)} = 0,  \qquad '=d/d\xi.
\end{align}
Integrating leads to the fourth order equation 
\begin{align}\label{eq:4th_ord_ode}
-c f + \frac{1}{2}f^2 +\alpha f''+  f^{(4)} = A, 
\end{align}
where $A$ is a real constant of integration. Multiplying Eq. \eqref{eq:4th_ord_ode} by $f$ and integrating once more leads to the Hamiltonian
\begin{align}\label{eq:energy_integral}
H=-\frac{c}{2}f^2 + \frac{1}{6}f^3 + \frac{\alpha}{2}(f')^2+f''' f'  - \frac{1}{2}\left(f''\right)^2 - A f, 
\end{align}
which is a second constant of integration, an invariant of solutions. The Hamiltonian structure is revealed by writing 
\beq\label{hamiltonian1}
 H(y,z,p,q)=pq-\frac{\alpha}{2}q^2-\frac12z^2 +\frac16 y^3-\frac{c}{2}y^2 - Ay. 
 \eeq
with 
$$y=f, \ z=f'',\ p=f'''+\alpha f',\ q=f'.$$
Then equation \eqref{eq:energy_integral} can be written as the first order system
\begin{align}\label{hamsys}
\begin{split}
y' &= \frac{\partial H}{\partial p}, \qquad   p' = -\frac{\partial H}{\partial y}\\
z' &=\frac{\partial H}{\partial q}, \qquad  q' =-\frac{\partial H}{\partial z}.
\end{split}
\end{align}
There are two scaling properties of traveling waves that are useful; these are enumerated in the following lemma. 
\begin{lemma}\label{Lemma_1}
\begin{enumerate}
\item  If $u(x,t)=f(x-ct)$ is a traveling wave solution of \eqref{eq:kawahara} with speed $c$, then for any $b\in \R,$ $f(x-\tilde{c}t) +b$ is also a solution, with speed $\tilde{c}=c+b.$
\item If $f(\xi)$ satisfies equation \eqref{eq:4th_ord_ode} with speed $c$ and constant of integration $A,$ then $\tilde{f}(\xi)=f(\xi)-d$ satisfies equation \eqref{eq:4th_ord_ode} with speed $\tilde{c}=-\sqrt{c^2+2A}$ and constant of integration $\tilde{A}=0,$ where $d=c+\sqrt{c^2+2A}.$ 
\end{enumerate}
\end{lemma}

\begin{proof}
Property 1. follows from the Galilean invariance property  \eqref{eq:galilean};  property 2. follows from substituting the formulas for $\tilde{f}, \tilde{c}$ into \eqref{eq:4th_ord_ode} with $A=0,$
or by substituting $\tilde{f}=f-d, \ \tilde{c}=c+e,$ and solving for $d$ and $e$ as functions of $A$ and $c$ to satisfy \eqref{eq:4th_ord_ode} with $A=0.$  
\end{proof}
  
\subsection{Background}
 
The Kawahara equation \eqref{eq:kawahara} and  the singular KdV5 equation, for which $\alpha=0,$ possess properties rather different from the canonical KdV equation
\begin{align}\label{eq:kdv}
u_t + uu_x + \alpha u_{xxx} = 0, \quad \alpha = \pm 1, 
\end{align}
which has  the soliton solution 
\begin{align}\label{eq:kdv_soli}
u(x,t) = \ubar + \alpha a \  {\rm sech}^2 \left(\sqrt{\frac{a}{12}}(x - c t -x_0)\right), \qquad c = \ubar + \alpha \frac{a}{3},
\end{align}
where $a>0$ is the soliton amplitude, $c$ is its velocity and $x_0$ is an arbitrary phase shift. The KdV soliton solution is a wave of elevation if $\alpha = +1$ and a wave of depression if $\alpha = -1$. 
Since the KdV equation is integrable, its solutions can be found using the inverse scattering transform \cite{ablowitz_solitons_1981}. However, the Kawahara equation is not integrable and a general characterization of solutions is not available.

Numerical computations of solitary wave solutions of Eq. \eqref{eq:kawahara} that asymptote to a constant $\ubar$ at infinity were first implemented by Kawahara \cite{kawahara_oscillatory_1972}, who observed that the structure of solutions depends on the choice of the parameter $\alpha$. With the normalization of Eq. \eqref{eq:kawahara} considered here, solitary wave solutions of Eq. \eqref{eq:kawahara} are waves of depression with a velocity less than $\ubar$. For $\alpha = 0$ or $\alpha = -1$, solitary wave solutions that decay to the constant $\ubar$ at $\infty$ exist for any velocity $c < \ubar$, whereas for $\alpha = +1$, solitary wave solutions can only be found for $c < -1/4 + \ubar$. Further details  of solitary wave solutions can be found in the  articles  \cite{amick_homoclinic_1992,champneys_homoclinic_1998}, along with an extensive list of relevant references pertaining to the existence and stability of these solutions. 

In addition to the solitary wave solutions that  decay exponentially to a constant at infinity, the Kawahara equation \eqref{eq:kawahara} also admits {\em generalized} solitary wave solutions when $\alpha = +1$. These solutions differ from a typical solitary wave in that they are not effectively localized, meaning  that instead of  decaying to a constant, a wave of elevation with a positive velocity relative to the background state is accompanied by co-propagating small oscillations that extend to infinity. The wavenumber of the small amplitude oscillations is selected by equating the elevated solitary wave velocity to the linear phase velocity \eqref{eq:phase1} \cite{benilov_generation_1993,grimshaw_weakly_1995,hunter_existence_1988}.

Multi-pulse solitary wave solutions of Eq. \eqref{eq:kawahara}, which resemble multiple separated copies of a single solitary wave have also been studied.  The first analytical work that explains multi-pulse solitary waves was introduced by Gorshkov and Ostrovsky \cite{gorshkov_interactions_1981} where the solitary waves are treated as independent particles with an interaction potential that depends on the length of separation between adjacent  pulses. Extrema of the interaction potential correspond to multi-pulse solutions. In fact, this method can be used to numerically compute an infinite family of multipulse solutions consisting of $N$ copies of the original solitary wave, naturally referred to as $N$-pulses. Using various techniques in dynamical systems theory, infinitely many $N$-pulse solutions of the ODE \eqref{eq:4th_ord_ode} have been shown to exist \cite{buffoni_global_1996,buffoni_bifurcation_1996,champneys_bifurcation_1993}.  A key feature in the construction of such solutions  is the computation of intersections of the stable and unstable manifolds of a hyperbolic equilibrium. A complementary method to construct these solutions uses a Lyapunov-Schmidt reduction and Lin's method \cite{sandstede_instability_1997}. Reference \cite{champneys_homoclinic_1998} provides an extensive review of the methodologies used to study homoclinic solutions for fourth-order reversible dynamical systems. In the case of 2-pulse solutions, rigorous analytical stability results are known \cite{chugunova_two-pulse_2007,parker_periodic_2020} and are complemented by early numerical simulations \cite{buryak_stability_1997}. The extensive collection of traveling wave solutions of the Kawahara and fifth order KdV equations can be largely attributed to the additional degrees of freedom inherent in the fourth order traveling wave ODE \eqref{eq:4th_ord_ode}. The Whitham modulation theory was also used in the description of unsteady dispersive shock wave solutions of the KdV5 and Kawahara equations in \cite{hoefer_modulation_2019,sprenger_shock_2017}. 
  
\subsection{Outline}\label{outline1}

The purpose of this manuscript is to study 
traveling wave solutions of the Kawahara equation \eqref{eq:kawahara} that either asymptote to different periodic waves at $\pm \infty$, or connect a periodic wave to a solitary wave. Such waves were identified for the KdV5 equation $(\alpha=0)$ in \cite{sprenger_discontinuous_2020}; the current paper extends those results and techniques to the Kawahara equation and introduces ideas from dynamical systems to clarify the structure of the solutions. In this subsection we outline the methods and procedure we use to investigate the traveling waves. 

The first key step is, in reality, an assumption about the existence of periodic traveling waves, though there is some theoretical justification for the existence of these solutions in the case of small amplitude waves using a Lyapunov-Schmidt argument\cite{haragus_spectral_2006}. Periodic traveling wave solutions, with speed $c,$ of   \eqref{eq:kawahara} will be associated with solutions $f(\xi),  \xi=x-ct$ of the ODE \eqref{eq:4th_ord_ode}. It is convenient to represent the period $2\pi/k$ explicitly, by setting $f(\xi)=\varphi(k\xi),$ where $k>0$ is the  wavenumber. Then $\varphi(\theta),  \theta=k\xi,$ is $2\pi-$periodic and satisfies
\begin{equation}\label{ode_1}
-c \varphi + \frac{1}{2}\varphi^2 +\alpha k^2\varphi''+  k^4\varphi^{(4)} = A,\quad  \varphi(\theta+2\pi)=\varphi(\theta), 
\end{equation}
where $'=d/d\theta,$ and $A$ is a constant of integration. Assumptions on the periodic traveling waves are summarized in the following. \\
 
\noindent (Assumption): (1) \ {\em For fixed constants $c,A,k,$ each $2\pi$-periodic solution  $\varphi(\theta)$ of \eqref{ode_1}  is associated with a  unique triple $(\ubar, a, k)$ representing the average $\ubar$ over a single period, the amplitude $a$ of the wave and it's wavenumber $k.$} 
(2) \ {\em Each triple $(\ubar, a, k)$ is associated with a $2\pi$-periodic solution of \eqref{ode_1}, in which the constants $c=c(\ubar, a, k)$ and $A=A(\ubar, a, k)$ are functions of the triple, and the solution $\varphi(\theta)$  is unique up to translation by a constant phase $\theta_0$.}\\

In practice, we approximate periodic traveling waves in two ways. The Stokes expansion is an asymptotic series that converges to the wave for small amplitude waves. For larger amplitudes, we rely on numerical methods to calculate a library of periodic solutions and their associated parameter triples. The procedure, summarized in Appendix A,  approximates the periodic solution with a truncated Fourier series with coefficients  computed by solving a nonlinear system of algebraic equations \cite{sprenger_discontinuous_2020,ehrnstrom_traveling_2009}. 

We seek traveling wave solutions   $u(x,t)=f(x-ct)$ of the Kawahara equation \eqref{eq:kawahara}, with fixed parameter $\alpha$, that uniformly approach distinct solutions, $f_\pm(\xi)$ as the traveling wave variable $\xi=x-ct$ approaches infinity: $|\xi| \gg 0$. The functions in the far-field are either periodic solutions of \eqref{eq:profile_eq}, or one is a  solitary wave. These two possibilities are studied separately, as solitary waves are  associated with the singular limit $k\to 0$ in \eqref{ode_1}.  

In the case of two periodic waves in the far field, they must have   velocities that match the TW velocity $c.$ If the triples of the periodic waves are denoted $(\ubar_\pm, a_\pm, k_\pm),$ then this implies they are related by the equation
\begin{equation}
\label{speeds1}
c(\ubar_+, a_+, k_+)=c(\ubar_-, a_-, k_-)= c .
\end{equation}
Equation \eqref{speeds1} is one of three equations, referred to as {\em jump conditions,} for the triples $(\ubar_\pm,a_\pm,k_\pm),$ the other two derived from equation \eqref{eq:4th_ord_ode} and the Hamiltonian \eqref{eq:energy_integral}.  

Similar jump conditions arise in the theory of shock solutions to hyperbolic conservation laws, and the corresponding bifurcation from trivial states is well established, using the shock speed as the bifurcation parameter.  
For our jump conditions, the trivial solution has equal triples, $(\ubar_-,a_-,k_-) = (0,a_+,k_+)$ (due to the Galilean invariance \eqref{eq:galilean}, we can fix $\ubar_+=0$ without loss of generality),  but the speed $c$ is not a suitable bifurcation parameter, as it is determined from either periodic solution. A convenient choice is to define the bifurcation parameter  $\zeta$ to be the average $\zeta = \half(k_++k_-),$ and  then introduce the parameter $\mu=\half(k_+-k_-).$ Instead of fixing the triple on one side, as is often done, we fix only the value $a_+=a$ of the amplitude on the right. Then for fixed $a \geq 0,$  the trivial solution is $(\ubar_-, a_-, \mu)=(0, a, 0)$, which satisfies the jump conditions for each $\zeta.$ 

The jump conditions are now three equations for $(\ubar_-, a_-, \mu)$ and $\zeta.$ Bifurcation points for a fixed value of $a$ are values of $\zeta$ (which is $k=k_+=k_-$ since $\mu=0$) for which the equations linearized about the trivial solution are singular. This corresponds to values of $\zeta$ for which the $3\times 3$ coefficient matrix $J(\zeta)$ is singular. Since the bifurcation parameter appears nonlinearly, $\zeta$ satisfies an equation with multiple solutions, not just the three in classical bifurcation theory for three equations, in which $\zeta$ would appear linearly.
Varying $a$ results in curves of bifurcation points given by functions  $\zeta =k_j(a),$ which follow the polynomial amplitude scaling \eqref{scaling1}  when $\alpha=0,$ but are more complicated for $\alpha=\pm 1.$ Bifurcation theory for nontrivial solutions proceeds as in the classical theory of Crandall and Rabinowitz \cite{crandall_bifurcation_1971}, since the bifurcation points are all simple ``eigenvalues", meaning the null space of the coefficient matrix is one-dimensional. The nondegeneracy conditions (non-zero values of certain coefficients) are verified numerically,  showing that each bifurcation point is located on a single curve of non-trivial solutions.  Each curve of non-trivial solutions bifurcates transcritically, meaning  that along the curve through the bifurcation point, $\zeta$ crosses $k_j(a).$ As in\cite{sprenger_discontinuous_2020}, it is convenient to parameterize each curve of nontrivial solutions with $\ubar_-,$ and to plot the curves as graphs of $k_\pm$ and $a_-$ as functions of $\ubar_-.$ Each point on a non-trivial branch is associated with a pair of periodic solutions. The curves terminate at values of $\ubar_-$ (depending on $a$) for which $k_+=0, $ or $k_-=0.$   

An alternative approach treats the parameters slightly differently, observing that the defining ODE for periodic solutions depends on two parameters $c$ and $A.$ Fixing these two parameters but allowing the wavenumber $k$ to vary then yields a one-parameter Hamiltonian $H(k;c,A)$. This scenario fits into the framework of Bridges and Donaldson \cite{bridges_degenerate_2005} where bifurcation of periodic orbits is related to critical points of $H(k;c,A).$ Since pairs of periodic solutions satisfying the jump conditions have equal values of $H,$ as well as $c$ and $A,$ the graph of $H(k;c,A)$ as a function of $k$ indicates suitable values of $k_\pm.$  
This approach gives pairs of triples $(k_\pm,c,A)$ depending on $H,$ treated as a parameter.

The phase space for the 4th order ODE \eqref{eq:4th_ord_ode}, 
written as a first order system, is the 4-dimensional space $\R^4$, with coordinates $(f,f',f'',f''').$ The Hamiltonian is constant on each solution, so that the level sets $H=constant$ are 3-dimensional hypersurfaces in $\R^4,$ and each periodic solution is a submanifold (topologically a circle) in one of these surfaces.  

Due to the Hamiltonian structure of the ODE, each periodic solution has a Floquet multiplier $\lambda=1, $ of  geometric multiplicity one and algebraic multiplicity 2, and two further multipliers $\lambda$ and $1/\lambda,$ with $\lambda\neq 1$ except at bifurcation points.   These values are also eigenvalues of the monodromy matrix.   When $\lambda>1,$ there is a two-dimensional unstable manifold consisting of solutions of the ODE asymptotic to the periodic solution as $\xi\to-\infty.$    Similarly there is a stable  manifold of solutions corresponding to the multiplier $1/\lambda.$  Interestingly, the multipliers can also be real and negative, or complex conjugates (hence lying on the unit circle in $\C$).
We refer to periodic solutions with real $|\lambda|>1$ as {\em nondegenerate.}

Now suppose we have a pair of   triples,  $(\ubar_\pm, a_\pm, k_\pm)$ satisfying the jump conditions, corresponding to periodic solutions $f_\pm=\varphi_\pm(k\xi)$ of the ODE  \eqref{eq:profile_eq} with the same speed. Consider the unstable
manifold of the periodic solution $f_-$  and the stable manifold of $f_+(\xi)$ at $\xi\sim +\infty$  These two-dimensional manifolds reside on
the same $H=constant$ 3D surface in $\R^4.$ Their intersection is then a curve, generically.  On a  Poincar\'{e} section $u=constant$
 the two manifolds are parameterized by the phases $\theta+$ and $\theta-$ of $\varphi_\pm(\theta)$ and the time-like
variable $\xi$. They each intersect the Poincar\'{e} section (which is 3-D) in a curve. But the $H=constant$ 3-D manifold intersects the
Poincare section in a 2-D surface, and both curves lie in that surface, so their intersection is at a point. When the periodic solutions are close, their values of u overlap, and the Poincar\'{e} section can be chosen with a value of u in the overlap. The two curves are then close
and generically they intersect transversally. This intersection point corresponds to a curve in the $\R^4$ phase
portrait, and is the trace of a TW joining the two periodic solutions. In principle, the curve can be obtained by setting $\xi=0$ at the intersection point, integrating back ($-\infty <\xi\leq 0$), i.e., along  the unstable manifold of $f_-$ and forward $0\leq \xi<\infty$) along the stable manifold of $f_+.$ In this construction of the TW, the two original phases $\theta_\pm$ become irrelevant, having been set by making the intersection point correspond to $\xi=0.$ In making this construction numerically, we encounter a problem corresponding to tracing a stable manifold towards a saddle point in two dimensions: any perturbation will veer off along the unstable manifold as the equilibrium is approached. In our case, increased  fidelity of the periodic asymptotic state can be achieved using refinement within the ODE solver.

This use of invariant manifolds  can be continued along each curve of non-trivial solutions of the jump conditions, for fixed parameter $a,$ as long as the intersection of the stable and unstable manifolds continues to be transversal. Eventually, one of the wavenumbers $k_\pm$ will approach zero, at which point the corresponding periodic orbit   approaches a heteroclinic orbit, and we find a portion of a solitary wave connecting to the other periodic solution.
 
The paper is organized as follows. In Section \ref{sec:periodic_orbits}, we describe a two-parameter family of periodic traveling wave solutions of the Kawawhara equation \eqref{eq:kawahara}. We recall the asymptotic approximation known as the {\em Stokes expansion}  for small amplitude waves, and  complement this analysis with 
numerical computations. 
Also in this section, we compute the Floquet multipliers of each  periodic solution, using the associated eigenvectors to generate   a two-dimensional stable manifold and a two-dimensional unstable manifold for each hyperbolic periodic orbit. 
In Section~\ref{jump_conditions1}, we derive necessary conditions for two periodic solutions to represent asymptotic (far-field) limits of a traveling wave solution of the Kawahara PDE \eqref{eq:kawahara}. These are derived by averaging the 4th-order ODE \eqref{eq:4th_ord_ode} and the Hamiltonian \eqref{eq:energy_integral} over each of the two periodic solutions. A crucial observation is that the   constant of integration $A$  and the   Hamiltonian $H$ must be the same on both periodic orbits, so they may be eliminated from the two averaged equations, by subtraction.  
The resulting pair of equations, together with the condition that both periodic solutions have the same wave speed, constitute a system of three equations referred to as {\em jump conditions.} They  are first analyzed in the weakly nonlinear limit using the Stokes approximation,   and then in the strongly nonlinear regime using the two-parameter library of numerically computed periodic solutions. In Section \ref{sec:heteroclinic_computations}, we complete the numerical construction of the  traveling wave solutions from pairs of periodic waves satisfying the jump conditions.  A similar construction is used to construct traveling  waves that connect a solitary wave  to a periodic wave.  
The paper concludes with a brief summary and discussion in Section~5.



\section{Periodic Traveling Wave Solutions of the Kawahara equation}
\label{sec:periodic_orbits} 

In this section, we describe periodic solutions of Eq. \eqref{eq:profile_eq} with fixed $\alpha,$ corresponding to periodic traveling wave solutions of \eqref{eq:kawahara}. We explore approximations of periodic traveling waves, in \S\ref{sec:periodic_orbits-computations}, their associated Floquet multipliers in \S\ref{floquet1}, and the connection to bifurcation points given by extrema of the Hamiltonian, in \S\ref{bifurcation_points}.

The existence of small amplitude, periodic traveling wave solutions to Eq. \eqref{eq:kawahara} was shown using a Lyapunov-Schmidt reduction \cite{haragus_spectral_2006}. In this same manuscript, the authors demonstrated that small amplitude waves with $\alpha = + 1$ are spectrally stable. Non-resonant, periodic traveling wave solutions have been shown to exhibit instabilities with respect to short-wavelength perturbations \cite{creedon_high-frequency_2021}. In the resonant case, periodic orbits can be approximated by a modification of the Stokes expansion \cite{haupt_modeling_1988}. A numerical study of the stability of resonant periodic solutions found  similar high-frequency instabilities \cite{trichtchenko_stability_2018}. Many studies of periodic solutions of Eq. \eqref{eq:kawahara} rely on asymptotic or numerical approximations of the periodic wave, though some closed form solutions do exist in terms of elliptic functions when $\alpha = -1$ \cite{mancas_traveling_2018}.

 \subsection{Approximate and numerical computations of periodic orbits}
\label{sec:periodic_orbits-computations}

We seek approximations of periodic solutions of \eqref{eq:profile_eq} by asymptotic analysis (for small amplitude waves) or with numerical methods (for larger amplitudes).
 As in \S \ref{outline1}, it is convenient to consider functions $\varphi(\theta)=f(\theta/k),$ where $f(\xi)$ is a periodic solution of  \eqref{eq:profile_eq} with wavenumber $k.$  
Then $\varphi(\theta)$ is a $2\pi$-periodic solution of Eq. \eqref{ode_1}. 
 These periodic solutions are parameterized by three constants:
 \begin{align}\label{parameters1}
\begin{split}
\text{wave mean:}  \qquad \ubar & = \frac{1}{2\pi}\int_0^{2\pi} \varphi(\theta) \ d\theta ,\\
\text{amplitude:}\qquad  a & = \max_\theta \varphi(\theta) - \min_\theta \varphi(\theta), \\ 
\text{wavenumber:} \qquad k.  &
\end{split}
\end{align}
We additionally assume that the wave speed $c$ is associated uniquely with the triple $(\ubar,a,k),$ referred to as  {\em wave parameters.} We thus have a three-parameter family of periodic traveling waves with associated  wave speeds:
$$
\varphi=\varphi(\theta;\ubar,a,k) ; \ \ c= c(\ubar,a,k).
$$
The Galilean invariance property \eqref{eq:galilean} is then expressed as 
\begin{align}\label{eq:shift_per}
\varphi(\theta;\ubar,a,k) = \ubar + \varphi(\theta;0,a,k), \qquad c(\ubar,a,k) = c(0,a,k) + \ubar.
\end{align}
This connection reduces the description of the three-parameter family to that of solutions $\varphi$ with zero mean $\ubar=0.$
 
\subsubsection{Stokes wave approximation}
Periodic traveling waves with a small amplitude $a>0$ can be approximated with the Stokes expansion. The Stokes-Poincar\'{e}-Lindstedt method relies on expanding both the wave profile $\varphi=\varphi(\theta;\ubar,a,k)$ and its frequency $\omega=\omega(\ubar,a,k)=kc(\ubar,a,k)$  in a power series in the small   parameter $a$
\begin{subequations}\label{eq:stokes_exp}
\begin{align}
\varphi &= \ubar + a \varphi_1(\theta) + a^2 \varphi_2(\theta) + \mathcal{O}(a^3) \\
\omega &= \omega_0 + a^2 \omega_2 + \mathcal{O}(a^3). 
\end{align} 
\end{subequations}
Here, coefficients $\varphi_n(\theta)$ of $a^n$ are $2\pi$ periodic functions to be determined such that 
$$
a^n\varphi_n(\theta) \ll a^m \varphi_m(\theta)
$$ 
for $n > m$ so that the asymptotic series is well-ordered. To leading order, the small parameter $a$ is the wave amplitude defined in \eqref{parameters1}. The linear dispersion  $\omega_0$ is defined in \eqref{eq:disp1},  and $a^2\omega_2$  is the lowest order nonlinear correction.  

Inserting the approximations \eqref{eq:stokes_exp} into the traveling wave ODE \eqref{ode_1}, and solving at increasing orders in $a$ yields the Stokes wave to $\mathcal{O}(a^2)$ 
\begin{align}\label{eq:stokeswave}
\varphi (\theta) &= \ubar + \frac{a}{2}\cos(\theta) - \frac{a^2}{240 k^4-48 \alpha  k^2} \cos 2\theta + \mathcal{O}(a^3)\\
\label{eq:stokes_freq}
\omega &= \ubar k  - \alpha k^3 + k^5 - \frac{a^2}{480 k^3-96 \alpha  k} + \mathcal{O}(a^3). 
\end{align}

The expansion is singular for $k^2=\alpha/5,$ which corresponds to the first wavenumber satisfying the resonance condition \eqref{eq:resonance}.  Near this resonant wavenumber, the Stokes expansion must be modified to include two harmonics of comparable magnitude at $\mathcal{O}(a)$. The appropriate modification of the Stokes expansion for resonant solutions of Eq. \eqref{ode_1} was developed by Haupt and Boyd \cite{haupt_modeling_1988}.  
However, this issue is beyond the scope of this paper, and does not affect our use of 
the leading order terms in the expansion. 

\subsubsection{Numerical computation of periodic orbits}\label{num_p}
Large amplitude periodic waves are approximated via a Fourier collocation method, the details of which can be found in Appendix~A. In general, each member of the three parameter family of periodic solutions is defined by the wave parameter triple $(\ubar,a,k)$ and has a corresponding velocity $c$. As mentioned previously in this section, we appeal to the Galilean transformation \eqref{eq:galilean} so that only periodic solutions with zero mean are computed. Moreover, the   amplitude scaling  \eqref{scale1}  for the KdV5 equation (in which $\alpha=0$ in Eq. \eqref{eq:kawahara}) further reduces the calculation   to  the single parameter family for which the amplitude of the periodic TWs can be taken to be $a = 1$ for which numerical calculations are reported in \cite{sprenger_discontinuous_2020}.

Periodic traveling wave solutions are computed to high accuracy for discrete values $(a_i,k_j),$ being careful about resonant waves when $\alpha = +1$. For $\alpha = -1$, $a_i = i \Delta a$ with $i = 1,2,\ldots N_a$ and $k_j = j\Delta k$ with $j = 1,2,\ldots, N_k$. We choose $N_a = 1600$ and $N_k = 300$ and $\Delta a = \Delta k = 0.005$. For $\alpha = +1$, we take care to avoid the small amplitude resonant waves that satisfy the condition \eqref{eq:resonance}. Consulting the expansion \eqref{eq:stokes_exp}, we observe a resonance at $k = \sqrt{1/5} \approx 0.447$.  Periodic waves are computed on two separate amplitude-wavenumber grids, but with the same spacing $\Delta a = \Delta k = 0.005$. For short wavelength waves,  we compute on the discrete grid  $(a_i,k_j)$ with $a_i = i \Delta a$ and $k_j = 0.445 + j \Delta k$ with $i = 1,2,\ldots,N_{a,1}$ and $j = 1,2,\ldots, N_{k,1}$. We choose $N_{a,1} = 1600$, $N_{k,1} = 211$ and a step size of $\Delta a = \Delta k = 0.005$. For longer wavelengths, we  compute on the discrete grid with $a_i = 0.65 + i \Delta a$ and $k_j = j \Delta k$ with $i = 1,2,\ldots,N_{a,2}$ and $j = 1,2,\ldots, N_{k,2}$ and  $N_{a,2} = 1471$ and $N_{k,2} = 89.$

In Figures \ref{fig:alpha_m1_per_waves} ($\alpha = -1$) and \ref{fig:alpha_p1_per_waves} ($\alpha = + 1$), the truncated Stokes wave profiles \eqref{eq:stokeswave} are compared to the corresponding numerically computed periodic waves. For larger wavenumbers, the Stokes expansions and the numerical profiles are indistinguishable, but for smaller wavenumbers, the Stokes expansion is a poor approximation.  Illustrative examples of this agreement are shown in Figs. \ref{fig:alpha_m1_per_waves}(a), \ref{fig:alpha_m1_per_waves}(b) \ref{fig:alpha_p1_per_waves}(a), and \ref{fig:alpha_p1_per_waves}(b). These comparisons between Stokes solution \eqref{eq:stokeswave} and the numerically computed periodic orbits suggest that weakly nonlinear approximations may provide helpful, but limited, information on the traveling wave solutions connecting two distinct periodic orbits. 
 
\begin{figure}[h]
\begin{center}
\includegraphics[scale=0.4,page = 1]{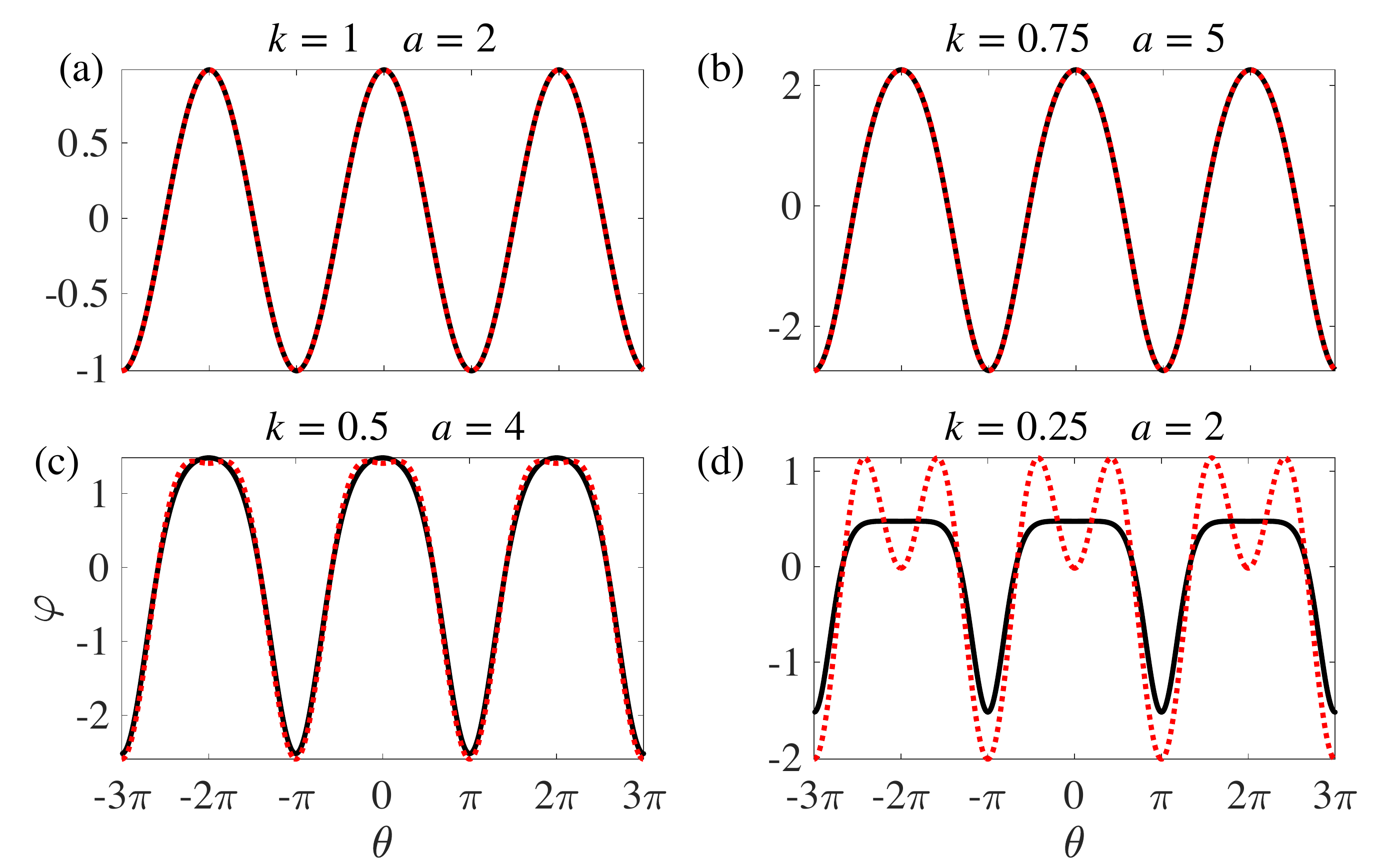}
\caption{($\alpha = -1.$) \ Comparison of Stokes expansions \eqref{eq:stokeswave} (dashed, red curves) to numerical solutions (solid, black curves).}
\label{fig:alpha_m1_per_waves}
\end{center}
\end{figure}

\begin{figure}[h]
\begin{center}
\includegraphics[scale=0.4,page = 2]{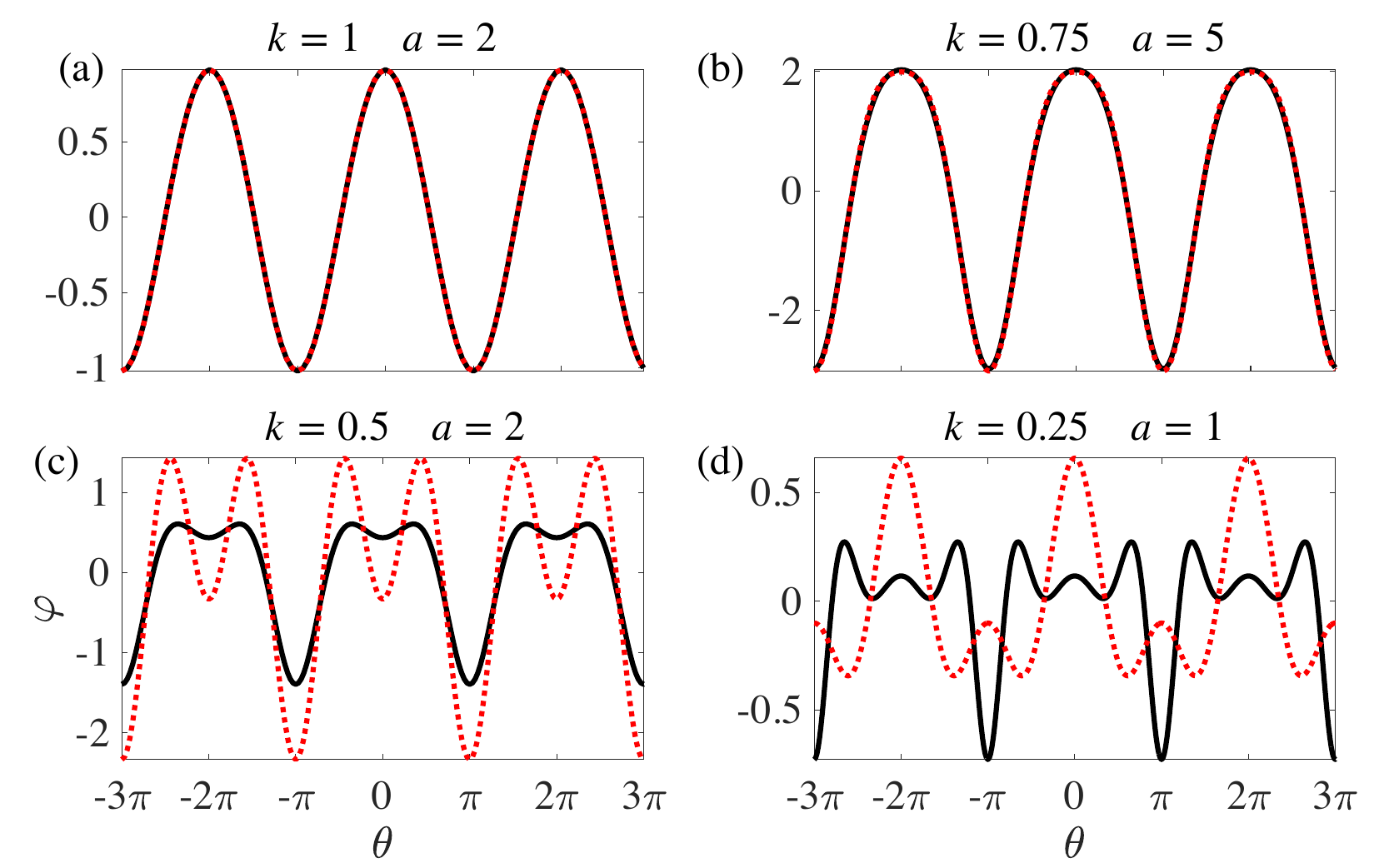}
\caption{($\alpha = +1.$) \ Comparison of Stokes expansions \eqref{eq:stokeswave} (dashed, red curves) to   numerical  solutions (solid, black curves).}
\label{fig:alpha_p1_per_waves}
\end{center}
\end{figure}

\subsection{Floquet multipliers}\label{floquet1}

The computation of Floquet multipliers for  periodic orbits is a key component of our construction of traveling waves. We begin by outlining the theoretical basis for the calculations. A thorough treatment of Floquet theory for Hamiltonian systems can be found in various texts \cite{koon_dynamical_2011,abraham_foundations_1987,wiggins_introduction_2003}.

Let $\varphi(\theta)$ be a $2\pi$-periodic solution of \eqref{ode_1}. The {\em flow map} $\Phi(\xi)$ is defined to be the linear map that governs the evolution of perturbations near the periodic orbit. It is given by the matrix-valued solution of the linear initial value problem 
\begin{equation}\label{eq:flow_map}
\frac{d\Phi}{d\xi} = \mathcal{B} \Phi, \quad \Phi(0) = I_4, 
\end{equation}
in which 
$$
\mathcal{B}(\xi)  = \left[\begin{array}{crrr}
0  & \ \ 1    &  \ \ 0   &  \ \ 0 \\ 
0 & 0 & 1 & 0 \\
0 & 0 & 0 & 1\\
c - \varphi(k\xi) & 0 & -\alpha & 0\\ 
\end{array}\right], $$
and $I_4$ denotes the $4 \times 4$ identity matrix. Evaluating $M=\Phi(2\pi/k)$ defines  the   {\em monodromy matrix,}  whose four eigenvalues $\left\{\lambda_i\right\}_{i = 1}^4$  are the {\em Floquet multipliers} of the periodic orbit  $\varphi.$ Since Eq. \eqref{eq:4th_ord_ode} has a Hamiltonian structure \eqref{hamiltonian1}, certain properties of the Floquet multipliers are guaranteed. The monodromy matrix, $M$, has an eigenvalue  $\lambda=+1$ with algebraic multiplicity two, and geometric multiplicity one, with eigenvector corresponding to $\varphi'(0)$. The remaining two {\em nontrivial} Floquet multipliers are reciprocals $\lambda, 1/\lambda$ of one another.  If a nontrivial Floquet multiplier $\lambda$  is real, the solution $\varphi$ is said to be {\em hyperbolic,} and if it is non-real (i.e.,  $|\lambda| = 1$ and $\lambda\neq \pm 1$), then  $\varphi$ is {\em elliptic}. These properties follow from the fact that the monodromy matrix of an autonomous Hamiltonian system is symplectic \cite{wiggins_introduction_2003}. 

To calculate $M$ and  the nontrivial Floquet multipliers,  we numerically integrate the ODE \eqref{eq:flow_map} for each periodic orbit $\varphi$ (computed in \S \ref{sec:periodic_orbits-computations}), with parameters $(k,a),$ and average $\ubar=0$.   In Figure \ref{fig:floquet_mult_config}, we illustrate five possible configurations of the Floquet multipliers $\lambda$, together with their algebraic multiplicity (when two or greater) in relation to the unit circle in $\C.$  

When the nontrivial Floquet multipliers $\lambda, 1/\lambda$ are real, as in Figs.\ref{fig:floquet_mult_config}(a,e), the periodic orbit has both a two-dimensional unstable manifold, and a two-dimensional stable manifold. The two dimensions come from the independent variable $\xi$ for the ODE \eqref{eq:4th_ord_ode} and the phase $\theta$ parameterizing points on the periodic orbit.  In the case of positive multipliers, the invariant manifolds are topologically cylinders, whereas for negative multipliers, they are topologically M\"o{}bius strips \cite{aougab_isolas_2019}. When Floquet multipliers lie on the unit circle, as in Figs. \ref{fig:floquet_mult_config}(b,c), there are no well-defined stable or unstable manifolds. However, non-real Floquet multipliers, as in Fig. \ref{fig:floquet_mult_config}(c), indicate that the periodic orbit is surrounded by an invariant torus, the consequences of which have been investigated numerically for  water wave models \cite{chen_numerical_1980,saffman_long_1980,vandenbroeck_new_1983,zufiria_weakly_1987}.

\begin{figure}[h] 
\begin{center}
\includegraphics[scale=0.5]{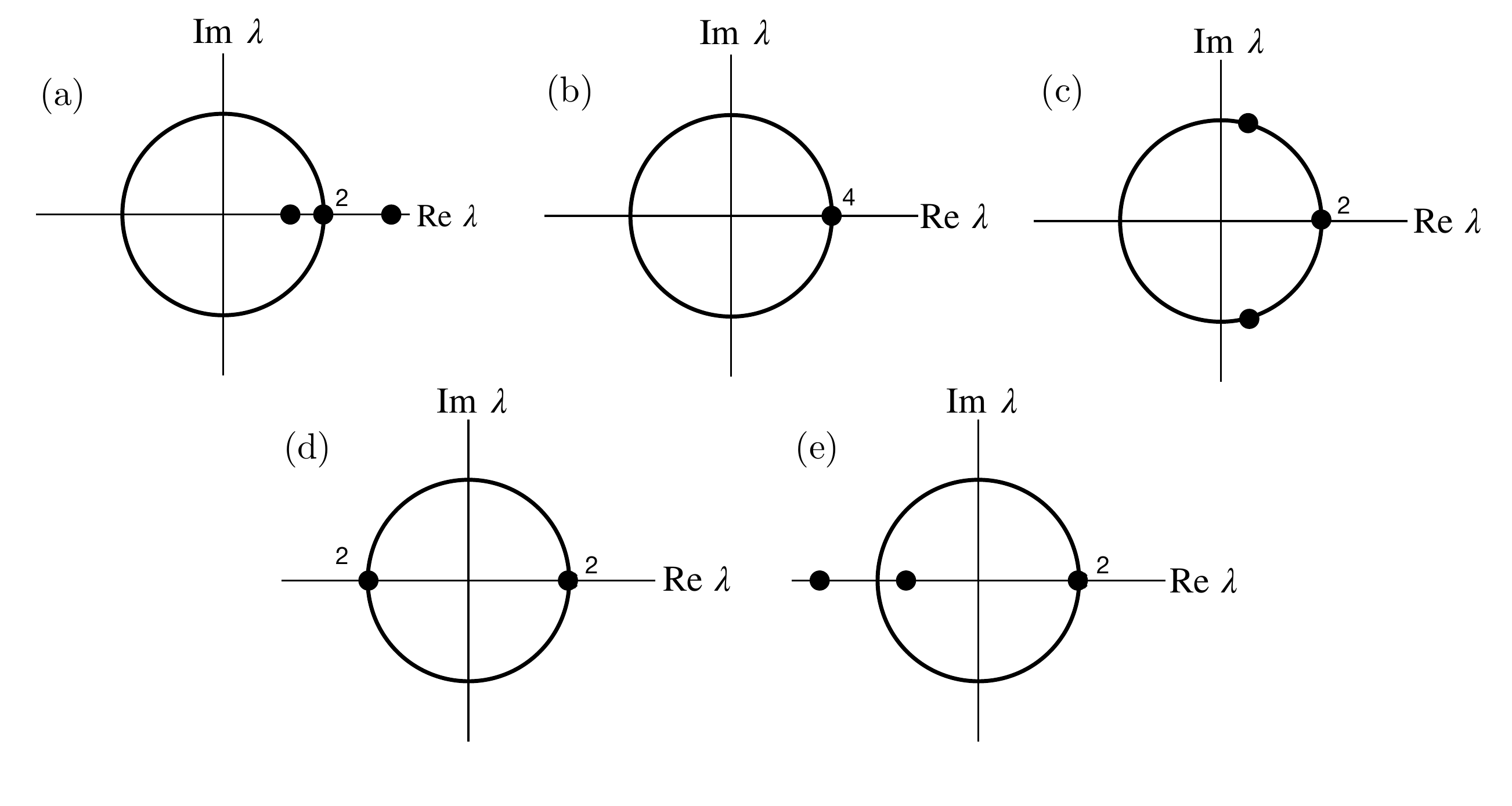}
\caption{Configuration of the Floquet multipliers (black dots) relative to the unit circle in the complex $\lambda$-plane.}
\label{fig:floquet_mult_config} 
\end{center}
\end{figure}

In Figure~\ref{fig:alpha_m1_floquet} we represent the results for $\alpha=\pm 1$ as solid red curves where the Floquet multipliers all coalesce at $\lambda=+1,$ and dashed blue curves, where the nontrivial multipliers coalesce at $\lambda=-1.$ Between the solid curves in Figure~\ref{fig:alpha_m1_floquet} 
there are regions identified by a ``$+$" symbol in the $(k,a)$ plane where the nontrivial multipliers are positive,  as in Figure \ref{fig:floquet_mult_config}(a). Regions bounded by dashed curves are marked by a ``$-$'' symbol; they indicate where nontrivial multipliers are negative, as in Figure \ref{fig:floquet_mult_config}(e). Regions filled with horizontal lines in Figure \ref{fig:alpha_m1_floquet} are bounded on one side by a dashed curve and on the other by a solid curve; they are where the Floquet multipliers are arranged as in Figure \ref{fig:floquet_mult_config}(c), with a complex conjugate pair on the circle between $+1$ and $-1.$ Such regions are typically so narrow that the dashed and solid curves are indistinguishable (and appear as red/blue curves).  In Figure \ref{fig:alpha_m1_floquet}(a), they are not at all visible, whereas in Figure \ref{fig:alpha_m1_floquet}(b) regions corresponding to elliptic periodic orbits are clearly visible only for larger values of $k.$

The limits of the $k-a$ parameter regime in  Figs. \ref{fig:alpha_m1_floquet}   are related to complications in the numerical calculations. Below $k = 0.2$  the largest Floquet multiplier typically grows well above $10^8$, so the condition number of the monodromy matrix is correspondingly very high  and accurate numerical computations of the Floquet multipliers necessitates higher than double numerical precision. In Figure \ref{fig:alpha_m1_floquet}(b), the computations are not continued into the gray region in the $(k,a)$ plane due to the presence of small amplitude resonant periodic solutions that were avoided in the numerical computations in \S \ref{sec:periodic_orbits-computations}. 

\begin{figure}[h]
\begin{center}
\includegraphics[scale=0.333]{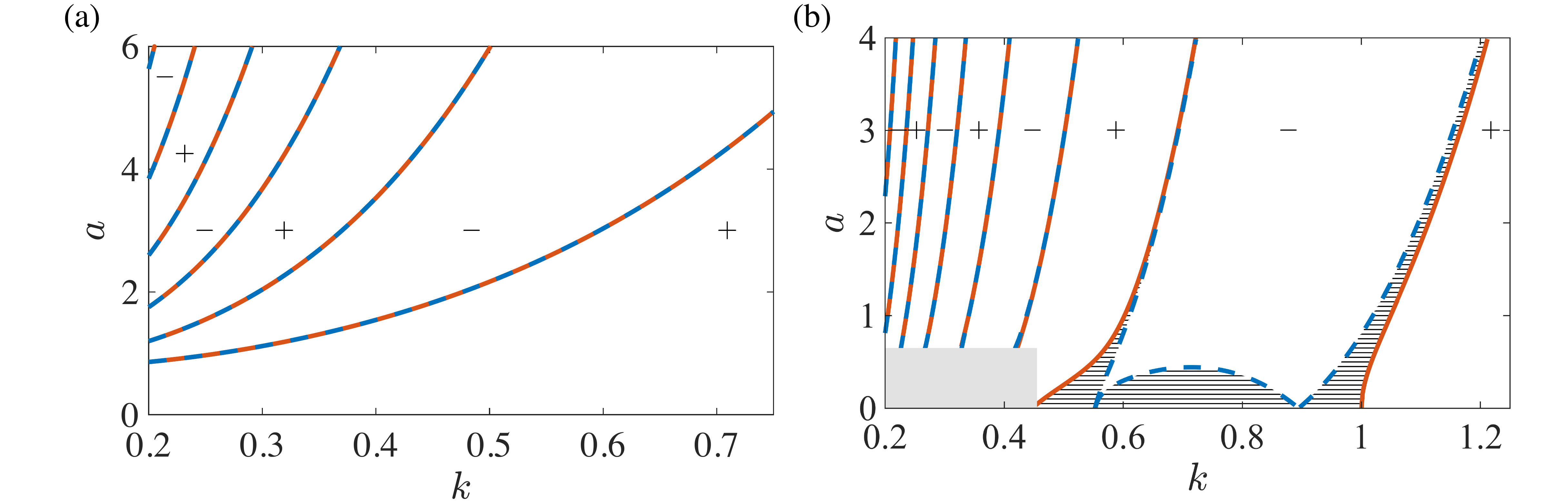}
\caption{Floquet multipliers corresponding to zero mean, periodic  traveling wave solutions of Eq. \eqref{eq:kawahara} with (a) $\alpha = -1$ and (b) $\alpha = +1$.  Dashed curves, corresponding to parameter values for which two Floquet multipliers coalesce at $-1,$ are indistinguishable from most of the solid curves, where all Floquet multipliers coalesce at $+1.$ The $\pm$  signs label regions bounded by either two solid curves, or two dashed curves. They indicate the sign of the nontrivial Floquet multipliers.}
\label{fig:alpha_m1_floquet}
\end{center}
\end{figure}

\subsection{Bifurcation of Traveling Waves}\label{bifurcation_points}
\label{sec:bif_TWs}

Observing Figure \ref{fig:alpha_m1_floquet}, we see numerous curves in the $(k,a)$ parameter space for which the corresponding $2\pi/k$-periodic wave has all Floquet multipliers coalescing at $+1,$ as in Figure~\ref{fig:floquet_mult_config}(b). In this section, we make the connection between these curves and bifurcation of periodic orbits at critical points of the Hamiltonian $H(k)$ as a function of wavenumber $k$ alone, keeping other parameters constant. 

For this discussion, we study a one-parameter family of periodic solutions of Eq. \eqref{ode_1} with a fixed wavespeed $c$ and constant of integration $A$. Using the Galilean transformation from Lemma~\ref{Lemma_1}, the constant of integration $A$ can be eliminated provided $c^2 + 2 A> 0,$ so we   set $A=0$ without loss of generality.   

A family of periodic solutions of  \eqref{ode_1} is computed for fixed $c$ and $A=0$, parameterized solely by their wavenumber $k$. Numerical computations of wave parameters $a, \ubar$ are given by functions of this sole parameter $a = a(k;c)$ and $\ubar = \ubar(k;c)$. In this scenario, the Hamiltonian, $H = H(k)$ in  \eqref{hamiltonian1} is a function of $k$ alone. The Hamiltonian is oscillatory for all velocities $c < -1/4$ if $\alpha = \pm 1$ \cite{chardard_stabilite_2009}. Illustrative examples of oscillatory $H$ are shown in Figure \ref{fig:hamiltonian_fixed_c} for fixed velocities $c$ chosen to accentuate oscillations. Each value of $k$ for which $H'(k) = 0$ corresponds to the four Floquet multipliers coalescing at $+1$ \cite{bridges_degenerate_2005}. In fact, there are an infinite number of such points accumulating as $k \to 0$ 
\cite{chardard_computing_2009}. 

From numerical calculations, we find that each such critical point of $H(k)$ is nondegenerate, so is a local extremum. Suppose $H(k)$ has a minimum at $k=k_0.$ Then for $H-H(k_0)>0$ but small, there are two values $k = k_\pm$ close to $k_0$ such that $H(k_+)=H(k_-)$ and $k_-<k_0<k_+.$  These values correspond to the wavenumbers of periodic orbits $\varphi_\pm$ that are close in the phase portrait, but have different averages $\ubar_\pm$ and amplitudes $a_\pm$ that depend on the wavenumbers $k_\pm$. From the way the Hamiltonians $H(k)$ are calculated, $\varphi_\pm$ satisfy the same ODE, i.e., with the same values of $c$ and $A=0$ as the periodic wave solution $\varphi_0$ with wavenumber $k=k_0.$  Numerical experiments indicate that for $\alpha = +1$, the Hamiltonian is oscillatory when the periodic wave velocities are $c \leq -1/4$ with constant of integration $A =0$. The Hamiltonian is similarly oscillatory for $\alpha = 0$, for any negative velocity with $A =0$.  When $\alpha = -1$, oscillations of the Hamiltonian only occur for $c < -1/4$; an example with $c = -1$ is shown in Fig. \ref{fig:hamiltonian_fixed_c}(a) with extrema of the Hamiltonian shown in the insets. The   function $H(k)$ is  monotonic for $\alpha = -1$ and $c = -1/4$; it  is shown in panel \ref{fig:hamiltonian_fixed_c}(b). These numerical experiments are consistent with similar  computations for the Kawahara equation in \cite{chardard_computing_2009}.   

In Figure~\ref{fig:hamiltonian_floquet_plot}(a) we plot the Hamiltonian near the critical point $k_0 \approx 0.6786$ together with the real part of Floquet multipliers $\lambda(k)$ of the periodic orbits over the same range of $k$. For this computation, we use $c = -1$, $A = 0$ and $\alpha = +1$. We see that $\varphi_-$ has real Floquet multipliers, whereas $\varphi_+$ has complex conjugate Floquet multipliers for $0.6786 \lesssim k_+ \lesssim 0.6802$.  This suggests there is an invariant torus associated with $\varphi_+,$ on which the structure of orbits may be quite complex. In particular, for rational values of ${\rm Arg}(\lambda), $ there are subharmonic solutions, either elliptic (with complex Floquet multipliers) or hyperbolic (with real Floquet multipliers). As $H$ decreases, the corresponding $\varphi_+$ regains real but negative Floquet multipliers, and now has well-defined stable and unstable manifolds, until $H$ nears a minimum of $H(k),$ where the multipliers again become complex conjugates, moving around the unit circle until they coalesce at $\lambda=1$ as the minimum is reached. Similar numerical computations near other extrema of the Hamiltonian show the same structure. That is,  for wave numbers on one side of the extremum,   the periodic waves have complex valued Floquet multipliers and on the other side the periodic waves are hyperbolic. This is a generic feature of the Hamiltonian for a one-parameter family of periodic solutions in a Hamiltonian system (see Theorem 5(iii) in  \cite{sepulchre_localized_1997}).

\begin{figure}[h]
\begin{center}
\includegraphics[scale=0.85]{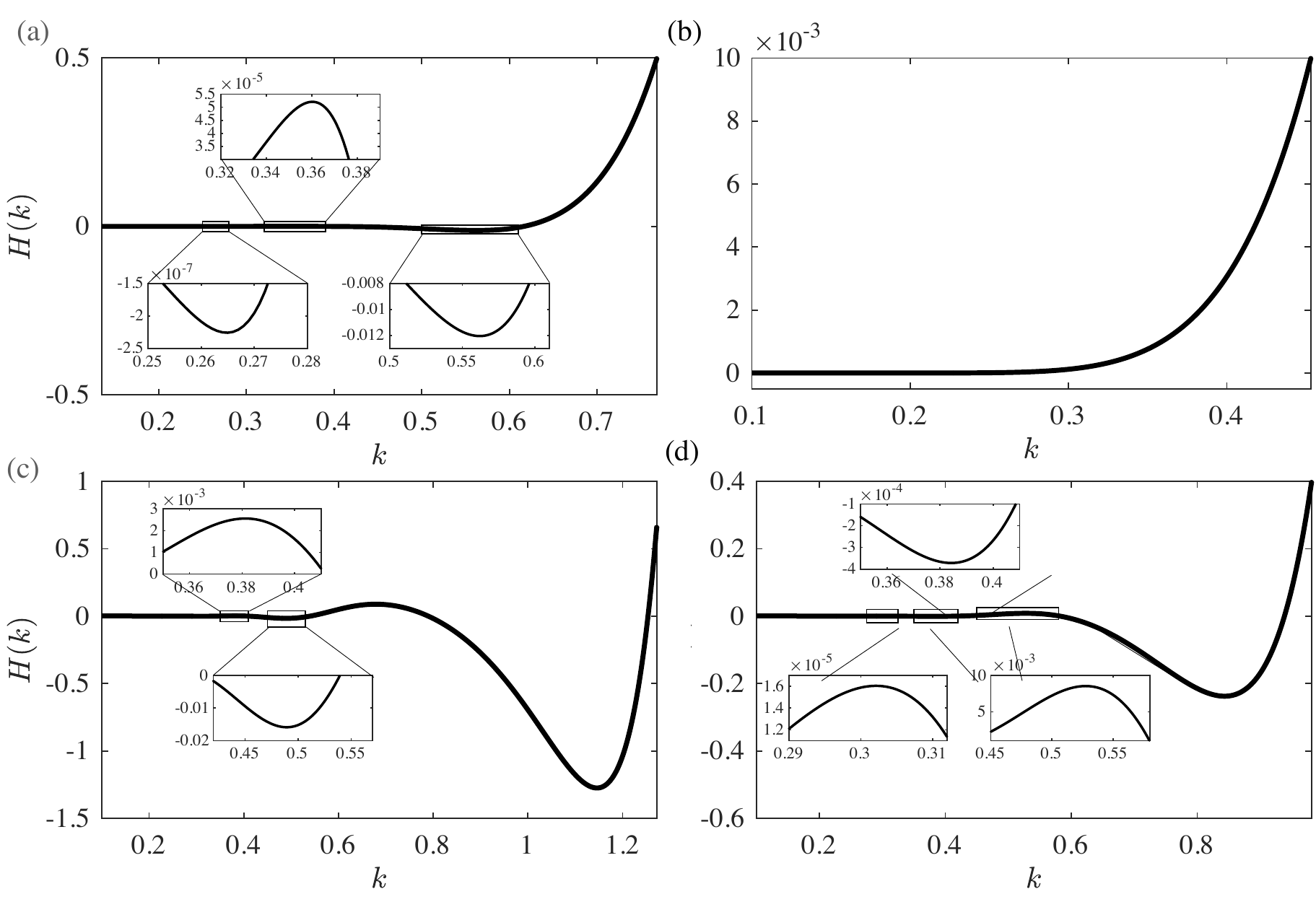}
\caption{Hamiltonian $H(k)$ 
(a) $\alpha = -1$, $c = -1$ (b) $\alpha = -1$ and $c = -0.25$ (c) $\alpha = +1$ and $c = -1$ (d) $\alpha = 0$ and $c = -1$ Insets are zoomed in short sections of the graph to show oscillations in cases where they are present.  }
\label{fig:hamiltonian_fixed_c}
\end{center}
\end{figure}

\begin{figure}[h]
\begin{center}
\includegraphics[scale=0.8]{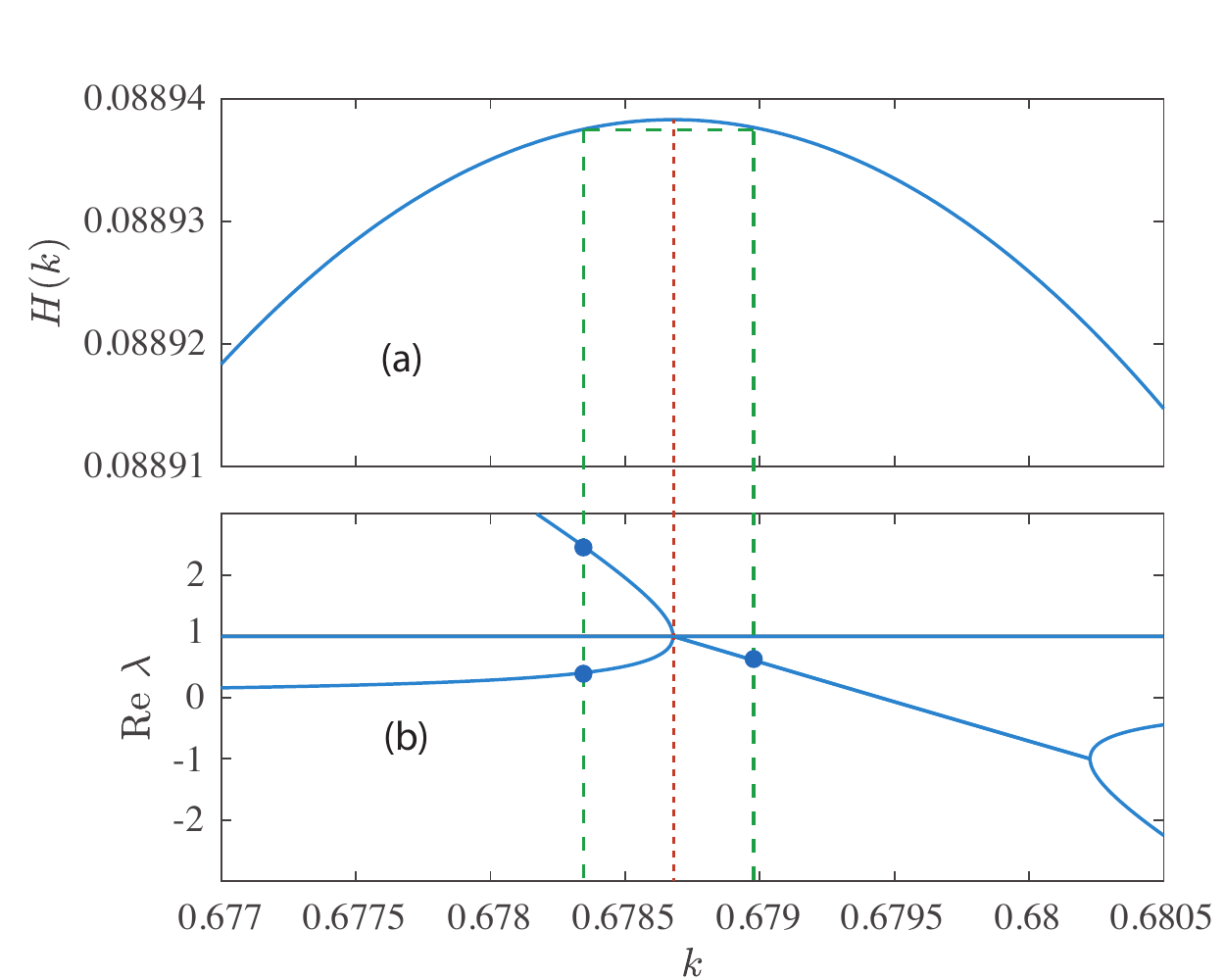}
\caption{($\alpha=+1, c=-1$) (a) $H(k)$ near a maximum at $k_0\sim 0.6786$ (b) Real part of Floquet multipliers $\lambda(k). $ The red dashed line indicates the critical point of $H(k).$ The green lines indicate equal values of $H(k)$ at $k_-<k_0$ and $k_+>k_0.$ }
\label{fig:hamiltonian_floquet_plot}
\end{center}
\end{figure}

  \section{Jump Conditions}\label{jump_conditions1}

In this section, we derive three conditions on pairs $f_\pm(\xi)$ of periodic solutions of the traveling wave ODE \eqref{eq:profile_eq} that are necessary conditions for a traveling wave solution of the PDE \eqref{eq:kawahara} which connects the two periodic waves. We relate these compatibility conditions to three jump conditions involving the corresponding constants $\ubar_\pm, a_\pm, k_\pm$ and their associated wave speeds $c(\ubar_\pm, a_\pm, k_\pm).$ The resulting equations for six unknowns are reduced using the scaling properties of subsection \ref{scaling1}, so that we are left with three equations in four unknowns. For small amplitude periodic solutions given by the Stokes expansion \eqref{eq:stokes_exp}, the system of equations is approximated by algebraic equations that can be solved explicitly to demonstrate the structure of solutions as a small set of curves in $\R^3.$ More generally, the system of equations is solved numerically, yielding multiple curves of solution. These computations illustrate that while the weakly nonlinear Stokes approximation provides valuable insight, it does not accurately capture the solutions of the jump conditions when the wavenumber of one of the periodic orbits is small. In the final subsection, we generalize the jump conditions to apply to traveling waves that connect a solitary wave with a given background $\ubar$ (in the singular limit $k\to 0$) to a periodic solution. 

Analysis of the jump conditions does not follow the pattern suggested by the theory of shock wave solutions of systems of conservation laws, in which characteristic speeds determine bifurcation points \cite{schaeffer_classification_1987} along a trivial solution, and the wave speed is an unknown to be determined along non-trivial branches. Instead, bifurcation points satisfy a nonlinear equation with many solutions, and the wave speed is determined by the parameters $\ubar, a, k$ of either periodic wave. 
 
 \subsection{Derivation of Jump Conditions}
 
 Recall that all traveling wave solutions $f(\xi), \xi=x-ct$ of the Kawahara equation \eqref{eq:kawahara} satisfy the pair of equations \eqref{eq:4th_ord_ode}, \eqref{hamiltonian1}, with constants $c,A,H.$ Now let $f(\xi)$ be a periodic traveling wave solution. We modify the averages \eqref{parameters1} of $2\pi$-periodic functions $\varphi$  to apply to $2\pi/k$-periodic functions. Averaging the traveling wave profile ODE \eqref{eq:4th_ord_ode} and the traveling wave Hamiltonian \eqref{hamiltonian1} over one wave period, we  arrive at 
 \begin{subequations}\label{averages1}
 \begin{align}
 -c \overline{f} +\frac 12  \overline{f^2}&=A,\\[6pt]
 \frac c2  \overline{f^2}-\frac 13  \overline{f^3}+\frac 32\alpha \overline{(f')^2}-\frac 52 \overline{(f'')^2}&=H, 
 \end{align}
 \end{subequations}
 where $\overline{F[f]} = \frac{k}{2\pi}\displaystyle\int_0^{2\pi/k} F[f(\xi)] d \xi$ is the average of the quantity $F[f]$ over the period of the $2\pi/k$-periodic function $f(\xi)$. To derive \eqref{averages1}(b) we have integrated by parts to simplify the expression.  
 
 Now let $f=f(x-ct)\in C^5(\R)$ be a traveling wave solution of \eqref{eq:kawahara}, such that  
 \beq\label{limits1}
 f(\xi)\to f_\pm(\xi) \quad \mbox{in} \ C^5  \  \mbox{as}\ \xi\to \pm\infty,
 \eeq
 for $C^5$ periodic solutions $f_\pm $ of  \eqref{eq:profile_eq}.
Then $f(\xi)$ also satisfies the ODE   \eqref{eq:profile_eq}, and hence $f, f_\pm$ are smooth $(C^\infty)$ and satisfy  the pair of equations \eqref{eq:4th_ord_ode}, \eqref{hamiltonian1}, for some  constants $c,A,H.$  Consequently, $f_\pm(\xi),$ associated with  parameters $\ubar_\pm, a_\pm, k_\pm,$ and speed $c(\ubar_\pm, a_\pm, k_\pm)$ respectively,  both satisfy \eqref{averages1} with the same speed $c.$  Thus, by eliminating $A$ and $H,$ we obtain the jump conditions
 \begin{subequations}\label{jump_conditions}
 \begin{align}\label{average1}
 -c (\overline{f_+}-\overline{f_-}) +\frac 12   \left(\overline{f_+^2} - \overline{f_-^2}\right)&=0,\\[6pt]
 \frac c2 \left( \overline{f_+^2}-\overline{f_-^2}\right)-\frac 13  \left(\overline{f_+^3}-\overline{f_-^3}\right)+\frac 32\alpha  \left(\overline{(f_+')^2}-\overline{(f_-')^2}\right)-\frac 52  \left(\overline{(f_+'')^2}-\overline{(f_-'')^2}\right)&=0, \label{hamil_conservation}\\[6pt]
 c(\ubar_+, a_+, k_+)-c(\ubar_-, a_-, k_-)&=0 \label{wave_conservation}, 
 \end{align}
 \end{subequations}
The first jump condition \eqref{average1} implies that both orbits $f_\pm$ individually satisfy the profile equation \eqref{eq:4th_ord_ode}, which is an obvious requirement of the traveling wave. The second jump condition \eqref{hamil_conservation} implies that the two periodic orbits lie on the same level set of the Hamiltonian, which is referred to as a wavenumber selection in  Remark I.1(b) of  \cite{knobloch_defectlike_2019}. The final jump condition \eqref{wave_conservation} expresses  $c(\ubar_\pm, a_\pm, k_\pm)=c.$ Note that all of the averages in these equations depend on the parameters $\ubar_\pm, a_\pm, k_\pm,$ but aside from $\overline{f_\pm}=\ubar_\pm,$ the other averages have to be evaluated to find their dependence on parameters. The connection between the jump conditions \eqref{jump_conditions} and traveling waves is summarized in the following.
 
 \begin{theorem}
 Suppose $f(\xi), \ -\infty<\xi<\infty$ is a   solution of the traveling wave ODE \eqref{eq:profile_eq}, corresponding to the traveling wave   with speed   $c,$ $u(x,t)=f(x-ct),$ solving the Kawahara equation \eqref{eq:kawahara}.  Let $f_\pm(\xi)$ be periodic solutions of  \eqref{eq:profile_eq} with the property that 
 $$
 f(\xi)\to f_\pm(\xi) \ \mbox{in} \ C^5 \  \mbox{as}\ \xi\to\pm\infty.
 $$
 Then $f_\pm$ satisfy the jump conditions \eqref{jump_conditions}.
 \end{theorem}
The proof, outlined above, is similar to that found in \cite{sprenger_discontinuous_2020}. 

\subsection{Solving the Jump Conditions}\label{jump_sols}
The jump conditions \eqref{jump_conditions} form a  system of three nonlinear equations in the six unknowns $\ubar_\pm, a_\pm, k_\pm,$ corresponding to the periodic solutions $f_\pm(\xi)$ of \eqref{eq:profile_eq}. There is a three-dimensional set of trivial solutions corresponding to $f_+=f_-,$ and we seek non-trivial solutions. To do so, we reduce the set of unknowns by first appealing to the Galilean symmetry \eqref{eq:galilean} of the Kawahara equation to take $\ubar_+=0.$ The set of five parameters $(\ubar_-,a_-,k_-,a_+,k_+)$ corresponding to non-trivial solutions of  the three equations \eqref{jump_conditions} with $\ubar_+=0$ consists of two-dimensional manifolds  in $\R^5.$ These surfaces bifurcate from the two-dimensional trivial solution 
$\ubar_-=0, a_+=a_-,$ and $k_+=k_-$, corresponding to $\varphi_+=\varphi_-.$ Let's consider parameterizing the two-dimensional surface of non-trivial solutions by $\ubar_-, a_+.$ Then the set of trivial solutions (for which $\ubar_-=0$) is represented by a plane $T$ in $\R^3:$
\begin{equation}\label{trivial1}
T=\{(a_-,k_-,k_+)=(a_+,k,k), a_+>0, k>0\}.
\end{equation}
 In what follows, we determine parts of the two-dimensional manifolds bifurcating from the trivial solution, beginning by determining curves in the plane of trivial solutions at which the Jacobian of the jump conditions, with respect to $(a_-,k_-,k_+)$ is singular. These curves contain  bifurcation points, at which the implicit function theorem fails to establish that the trivial solutions are the only nearby solutions. We begin out study of the jump conditions in the weakly nonlinear limit, where the averages in the jump conditions are explicit, and then utilize the library of numerically computed solutions to find branches of nontrivial solutions. 

\subsubsection{Weakly nonlinear regime}
\label{sec:bif_points}

We begin by studying the jump conditions in an explicit form in which the left and right periodic orbits are well approximated by the Stokes wave approximation \eqref{eq:stokeswave}.  In this weakly nonlinear regime, $a \ll 1$, 
 averages can be expressed in powers of the amplitude parameter $a$ and wavenumber $k$ 
\begin{align*}
    \overline{\varphi^2} &= \ubar^2 + \frac{a^2}{8} + \ldots \qquad 
   \overline{\varphi^3} = \ubar^3 + \frac{3}{8} \ubar a^2 +\ldots  \qquad
    \overline{\varphi_\theta^2} = \frac{a^2}{8} + \ldots \qquad
    \overline{\varphi_{\theta\theta}^2}  = \frac{a^2}{8} + \ldots \\
\end{align*}
Inserting these expansions into the jump conditions \eqref{jump_conditions}, setting $\ubar_+=0,$ and retaining terms up to  $\mathcal(O)(a^2)$ we have
 {
\begin{subequations}\label{eq:weaklyNL_jump_conds}
\begin{align}
-c_+ \ubar_- + \left(\frac{1}{2}\ubar_-^2 + \frac{a_-^2}{16} - \frac{a_+^2}{16} \right) & = 0 
\\
-c_+ \left(\frac{1}{2}\ubar_-^2 + \frac{a_-^2}{16} - \frac{a_+^2}{16}\right) + \left(\frac{1}{3}\ubar_-^3 + \frac{1	}{8}\ubar_-a_-^2 -\frac{3}{16}\alpha a_-^2 k_-^2 + \frac{5}{16}k_-^4 a_-^2  + \frac{3}{16}\alpha a_+^2 k_+^2 - \frac{5}{16}k_+^4 a_+^2 \right) &= 0 
\\
c_+-c_- & = 0, 
\end{align}
\end{subequations}
}
where $c_-$ and $c_+$ are the  velocities of the periodic orbits in the far-field given in  \eqref{eq:stokeswave}
\begin{equation}\label{speeds2}
c_\pm= \ubar_\pm  - \alpha k_\pm^2 + k_\pm^4 - \frac{a_\pm^2}{96 k_\pm^2(5k_\pm^2- \alpha )} + \mathcal{O}(a^3). 
\end{equation}
With the change of variables, 
\begin{align}\label{eq:COV}
\mu &=\frac 12(k_- - k_+) \quad \zeta  = \frac 12(k_- + k_+), 
\end{align}
the trivial solution becomes $u_- = 0, \ a_- = a_+,  \ \mu=0,$ for arbitrary values of the bifurcation parameter $\zeta.$ Then bifurcation points are values of $\zeta$ for which the equations, linearized with respect to 
the trivial solution, are singular. These points will depend on $a_+,$ so we obtain curves of bifurcation points $\zeta=\zeta_j(a_+).$  
The  jump conditions \eqref{eq:weaklyNL_jump_conds} are now
 
{\footnotesize
\begin{subequations}\label{eq:jump_conds_newvars}
\begin{align}
-c_+ \left(\ubar_-\right) + \left(\frac{1}{2}\ubar_-^2 + \frac{a_-^2}{16} - \frac{a_+^2}{16} \right) & = 0\\
-c_+\left(\frac{1}{2}\ubar_-^2 + \frac{a_-^2}{16} - \frac{a_+^2}{16}\right) + \frac{1}{3}\ubar_-^3 + \frac{1	}{8}\ubar_-a_-^2 -\frac{3}{16} \alpha  a_-^2 (\zeta +\mu )^2+\frac{5}{16} a_-^2 (\zeta +\mu )^4+\frac{3}{16} \alpha  a_+^2 (\zeta -\mu )^2-\frac{5}{16} a_+^2 (\zeta -\mu )^4 &= 0 \\
c_+ - c_- &= 0,  
\end{align}
\end{subequations}
}
 
The jump conditions are now three equations for the variables $\mathbf{z} = (\ubar_-,a_-,\mu)$ that take the form $G(\mathbf{z};a_+,\zeta) = 0$, where $G: \mathbb{R}^3 \times \mathbb{R}^2 \to \mathbb{R}^3$. We seek bifurcations from trivial solutions for a fixed value of $a_+$. Here, bifurcation points are values of $\zeta$ for which the equations linearized about the trivial solution $\mathbf{z}_0 = [0,a_+,0]$ are singular. Since $\mu = 0$, bifurcation points are identified by the wavenumber parameter $k_- = k_+ = \zeta$. Bifurcation points of the nonlinear system \eqref{eq:weaklyNL_jump_conds} therefore correspond to values of $\zeta$ for which the $3 \times 3$ Jacobian matrix, $dG(\mathbf{z}_0;a_+,\zeta)$  is singular. We treat $\zeta = \frac 12(k_- + k_+)$ as a bifurcation parameter, which appears nonlinearly. In the related scenario, where jump conditions arise in the context of discontinuous shock solutions of systems of conservation laws, the bifurcation parameter is the velocity, which appears linearly and is an eigenvalue of the jump conditions linearized about the trivial solution. For a $3\times 3$ system,  this then allows three distinct bifurcation points to be identified. The nonlinear dependence on the bifurcation parameter here could in principle result in more (or fewer) bifurcation points for fixed $a_+$. Varying the parameter $a_+$ results in curves of bifurcation points, $\zeta = k_j(a_+)$. These curves obey the amplitude scaling \eqref{scale1} when $\alpha = 0$, but are more complicated when $\alpha = \pm 1$.

In Figure~\ref{fig:bif_points_weak_NL}(a)-(c) ($\alpha=\pm1, 0$), we plot values of the pair $(\zeta=k,a)$ corresponding to zero-mean periodic waves at which the linearized jump conditions are singular, indicating  bifurcation to nontrivial solutions of the jump conditions \eqref{eq:jump_conds_newvars}. Figure~\ref{fig:bif_points_weak_NL}(a) ($\alpha = -1$) suggests that for sufficiently large   amplitudes, there are three distinct branches of nontrivial shock solutions,   two of which coalesce at a finite amplitude. In Figure~\ref{fig:bif_points_weak_NL}(b) ($\alpha = +1$), there are four distinct curves of bifurcation points, and a fifth appears for long waves   at small amplitudes. The four curves coalescing at $a = 0$ meet at $k = \sqrt{1/5}$, which is the first wavenumber that satisfies the resonance condition \eqref{eq:resonance}. Finally, in Figure~\ref{fig:bif_points_weak_NL}(c) ($\alpha = 0$), there are three bifurcation curves, on each of which $a$ is proportional to $k_\pm^4$, due to the scale invariance \eqref{scale1} of the KdV5 equation. The   bifurcation curves for $\alpha=0$ are consistent with previous results \cite{sprenger_discontinuous_2020}.

 In Figure \ref{fig:bif_points_weak_NL}(d-f), we plot example solutions of equations \eqref{eq:jump_conds_newvars} for sample fixed values of $a_+$ indicated on the  curves of bifurcation points in Figure \ref{fig:bif_points_weak_NL}(a-c). Since the Stokes expansion is only valid for $0 < a \ll 1$, then the solutions of the jump conditions may be misleading for the values of $a_+$ used here, which are of $\mathcal{O}(1)$. In the following section, we compute the bifurcation points and nontrivial solutions of the jump conditions found using the computed library of numerical solutions so that we may directly compare.
\begin{figure}[h!]
\begin{center}
\includegraphics[scale=0.15,page = 1]{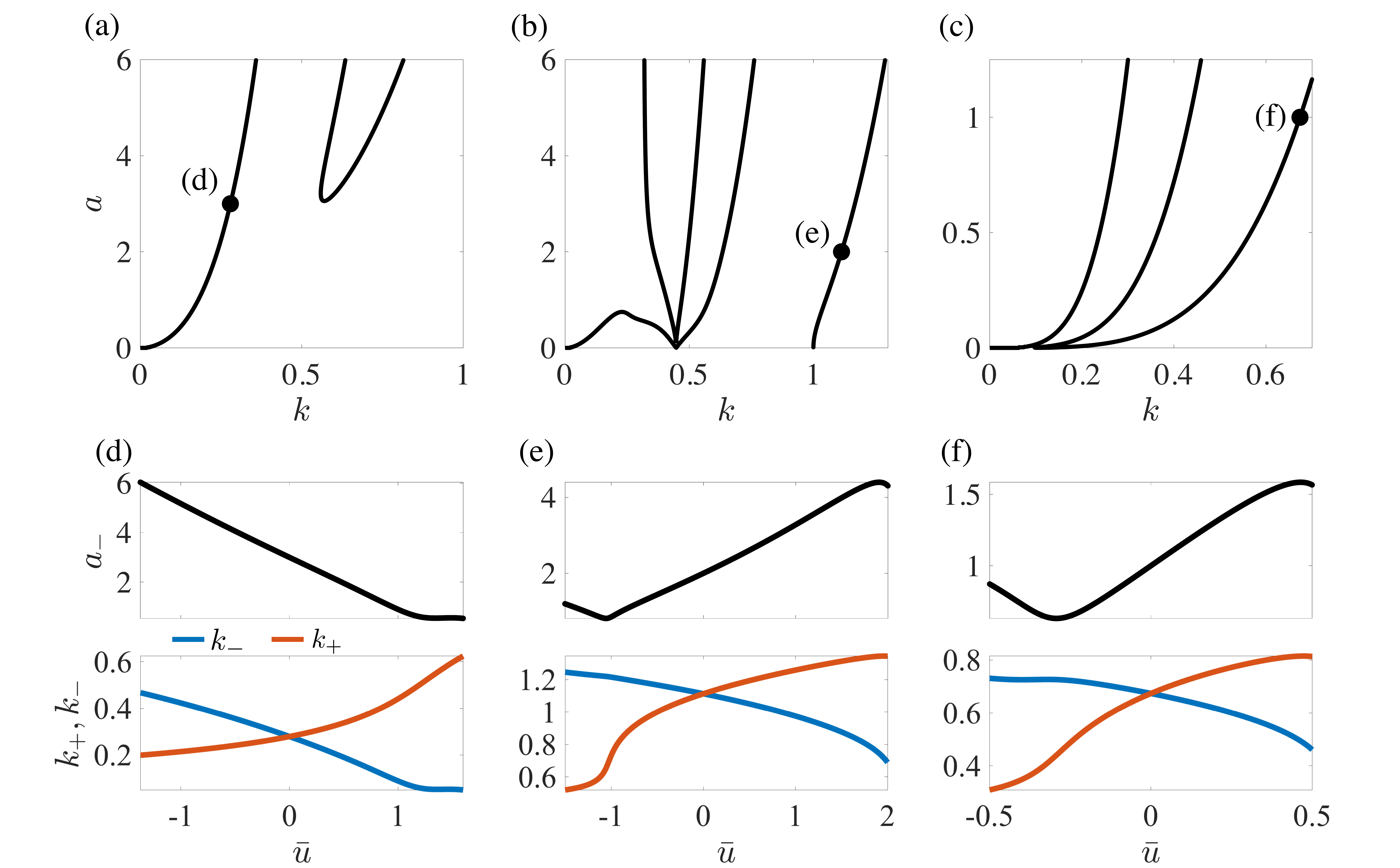}
\caption{Curves  of   bifurcation points of the weakly nonlinear jump conditions \eqref{eq:jump_conds_newvars} with (a) $\alpha = -1$, $a_+= 3$, (b) $\alpha = +1$, $a_+ = 2$, (c) $\alpha = 0$, $a_+ = 1$. (d) - (f):  nontrivial solutions of the jump conditions bifurcating from bifurcation points labeled (d)-(f).}
\label{fig:bif_points_weak_NL}
\end{center}
\end{figure}
Although the bifurcation points and curves of Figure~\ref{fig:bif_points_weak_NL}  are readily calculated, they contain only limited information, that becomes inaccurate for larger values of $a$ or smaller values of $k.$ However, these calculations demonstrate the connection between curves of parameters  representing bifurcation points, and the curves of parameters representing periodic solutions satisfying the jump conditions, if only approximately in the weakly nonlinear regime.

\subsubsection{Numerical computation of solutions of the full jump conditions}
\label{sec:numerical_shock} 
  Recall that $f(\xi)=\varphi(k\xi)$ in \eqref{jump_conditions}  and we take $\ubar_+=0$ without loss of generality. With the change of variables \eqref{eq:COV} the full jump conditions \eqref{jump_conditions}  become
{\scriptsize\begin{subequations}\label{eq:jump_conds_mu_zeta}
\begin{align}
-c \left(\ubar_-\right) +  \left(\frac{1}{2}\overline{\varphi_-^2}-\frac{1}{2}\overline{\varphi_+^2}\right) & = 0,\\
    -\frac{c}{2}\left(\overline{\varphi_-^2} - \overline{\varphi_+^2} \right) +
    \frac{1}{3}\left(\overline{\varphi_-^3} - \overline{\varphi_+^3} \right) - \alpha \frac{3}{2}\left((\zeta+\mu)^2 \overline{(\varphi_{-,\theta})^2} - (\zeta-\mu) ^2\overline{(\varphi_{+,\theta})'^2}\right)
  +   \frac{5}{2}\left((\zeta + \mu)^4 \overline{(\varphi_{-,\theta\theta})^2} - (\zeta-\mu)^4
    \overline{(\varphi_{+,\theta\theta})^2} \right) & = 0,\\
c_+ - c_- & = 0, 
\end{align}
\end{subequations}}

\noindent where $c_\pm=c(\ubar_\pm,a_\pm, k_\pm)=c.$ Here, we have used the change of variables in the integral, so that we average powers of $2\pi$-periodic functions $\varphi_\pm(\theta;\ubar_\pm,a_\pm,k_\pm)$ and their derivatives.  The jump conditions are trivially satisfied for any value of $a_+$ and $\zeta$ if $(\ubar_-,a_-,\mu) = (0,a_+,0),$ the trivial solution. As in the previous section, we find curves of bifurcation points at  values of the parameters $a_+, \zeta$ at which the Jacobian matrix of the nonlinear system \eqref{eq:jump_conds_mu_zeta} is singular, determined  by taking gradients of the equations with respect to the parameters $(\ubar_-,a_-,\mu)$ at the trivial soluton. 
 
Differentiating and evaluating at the trivial solution gives 
\begin{align}
\begin{split}
d\mathcal{G} &= \begin{bmatrix}
-c_+  & \half \frac{\partial \overline{\varphi_-^2}}{\partial a_-} & \half \frac{\partial \overline{\varphi_-^2}}{\partial \mu } \\
g_{21} & g_{22} & g_{23}
\\
-1 & \frac{\partial c_-}{\partial a_-} & \frac{\partial c_+}{\partial \mu}-\frac{\partial c_-}{\partial \mu}
\end{bmatrix}, \\ 
g_{21} & = -\frac{c_+}{2} \frac{\partial \overline{\varphi_-^2}}{\partial \ubar_-} + \frac{1}{3}\frac{\partial \overline{\varphi_-^3}}{\partial \ubar_-} \\ 
g_{22} & = -\frac{c_+}{2} \frac{\partial \overline{\varphi_-^2}}{\partial a_-} + \frac{1}{3}\frac{\partial \overline{\varphi_-^3}}{\partial a_-} - \alpha \frac{3}{2} \zeta^2 \frac{\partial \overline{\varphi_{-,\theta}^2}}{\partial a_-} + \frac{5}{2} \zeta^4\frac{\partial \overline{\varphi_{-,\theta\theta}^2}}{\partial a_-}\\ 
g_{23} & = - \frac{c_+}{2} \left(\frac{\partial \overline{\varphi_-^2}}{\partial \mu}-\frac{\partial \overline{\varphi_+^2}}{\partial \mu}\right) + \frac{1}{3}\left(\frac{\partial \overline{\varphi_-^3}}{\partial \mu}-\frac{\partial \overline{\varphi_+^3}}{\partial \mu}\right)- \alpha \frac{3}{2} \zeta^2\left(\frac{\partial \overline{\varphi_{-,\theta}^2}}{\partial \mu}-\frac{\partial \overline{\varphi_{+,\theta}^2}}{\partial \mu} \right)\\ & \quad  + \frac{5}{2} \zeta^4 \left(\frac{\partial \overline{\varphi_{-,\theta\theta}^2}}{\partial \mu}-\frac{\partial \overline{\varphi_{+,\theta\theta}^2}}{\partial \mu} \right)
\end{split}
\end{align}
The Jacobian is simplified since $\varphi_- = \varphi_+$ at the trivial solution. We can  represent the single periodic orbit in the far-field by $\tilde{\varphi}$ with parameters $\tilde{\ubar} = 0$, $\tilde{a} = a_+$, and $\tilde{k} = \zeta$ and the corresponding velocity $\tilde{c}$. Gradients are then taken with respect to the wave parameters $\tilde{\ubar}$ and $\tilde{a}$
\begin{align*}
\frac{\partial}{\partial \ubar_-} \to \frac{\partial}{\partial \tilde{\ubar}}, \qquad 
\frac{\partial}{\partial a_-} \to \frac{\partial}{\partial \tilde{a}}, 
\end{align*}
and gradients with respect to $\mu$ become
\begin{align*}
\frac{\partial}{\partial \mu} \overline{\left(\partial^n_\theta \varphi_-\right)^m} - \frac{\partial}{\partial \mu} \overline{\left(\partial^n_\theta \varphi_+\right)^m} \to 2 \frac{\partial}{\partial \tilde{k}} \overline{\left(\tilde{\partial}^n_\theta \varphi_+\right)^m}. 
\end{align*}  
The Jacobian at the trivial solution is then 
\begin{align}\label{eq:jacobian_new_vars}
\begin{split}
d\mathcal{G} &= \begin{bmatrix}
-\tilde{c}  & \half \frac{\partial \overline{\tilde{\varphi}^2}}{\partial \tilde{a}} & \half \frac{\partial \overline{\tilde{\varphi}^2}}{\partial \tilde{k}} \\
\tilde{g}_{21} & \tilde{g}_{22} & \tilde{g}_{23}
\\
-1 & \frac{\partial \tilde{c}}{\partial \tilde{a}} & 2\frac{\partial \tilde{c}}{\partial \tilde{k}}
\end{bmatrix}, \\ 
\tilde{g}_{21}  &= -\frac{\tilde{c}}{2} \frac{\partial \overline{\tilde{\varphi}^2}}{\partial \tilde{\ubar}} + \frac{1}{3}\frac{\partial \overline{\tilde{\varphi}^3}}{\partial \tilde{\ubar}},  \\ 
\ \ \tilde{g}_{22} & = -\frac{\tilde{c}}{2} \frac{\partial \overline{\tilde{\varphi}^2}}{\partial \tilde{a}} + \frac{1}{3}\frac{\partial \overline{\tilde{\varphi}^3}}{\partial \tilde{a}} - \alpha \frac{3}{2} \tilde{k}^2 \frac{\partial \overline{\tilde{\varphi}_{\theta}^2}}{\partial \tilde{a}}\nonumber + \frac{5}{2} \tilde{k}^4\frac{\partial \overline{\tilde{\varphi}_{\theta\theta}^2}}{\partial \tilde{a}} \\ 
\tilde{g}_{23} & = - \tilde{c} \frac{\partial \overline{\tilde{\varphi}^2}}{\partial \tilde{k}} + \frac{2}{3}\left(\frac{\partial \overline{\tilde{\varphi}^3}}{\partial \tilde{k}}\right)- 3 \alpha \tilde{k}^2 \frac{\partial \overline{\tilde{\varphi}_{\theta}^2}}{\partial \tilde{k}} + 5 \tilde{k}^4 \frac{\partial \overline{\tilde{\varphi}_{\theta\theta}^2}}{\partial \tilde{k}}. 
\end{split}
\end{align}
For fixed values of $\tilde{a}$, values of $\tilde{k}$ are computed for which which the Jacobian \eqref{eq:jacobian_new_vars} is singular. To compute such points, approximations of gradients with respect to parameters $\tilde{\ubar}$, $\tilde{a}$, and $\tilde{k}$ are required. Gradients with respect to $\tilde{\ubar}$ are computed explicitly using the identities 
\begin{align*}
\overline{\tilde{\varphi}^2} &= \tilde{\ubar}^2 + \overline{\varphi(\theta;0,\tilde{a},\tilde{k})^2},
&\overline{\tilde{\varphi}^3} &= \tilde{\ubar}^3 + 3\tilde{\ubar}\overline{\varphi(\theta;0,\tilde{a},\tilde{k})^2} + \overline{\varphi(\theta;0,\tilde{a},\tilde{k})^3} ,\\ 
\overline{\tilde{\varphi}_\theta^2} &= \overline{\varphi_\theta(\theta;0,\tilde{a},\tilde{k})^2}, 
&\overline{\tilde{\varphi}_{\theta\theta}^2}& = \overline{\varphi_{\theta\theta}(\theta;0,\tilde{a},\tilde{k})^2},
\end{align*} 
while gradients with respect to $\tilde{a}$ and $\tilde{k}$, on the other hand, must be approximated numerically. To do so, eighth order finite difference formulas are applied to the relevant averages of $\tilde{\varphi},$ computed using the library of 
periodic solutions discussed in \S \ref{sec:periodic_orbits-computations}. The coefficients of the finite difference formulas can be computed rapidly for the high-order finite difference schemes \cite{fornberg_generation_1988}.  Bifurcation points from trivial solutions of the jump conditions \eqref{jump_conditions} are then identified by parameter values for which the Jacobian matrix \eqref{eq:jacobian_new_vars} is singular.  In Figs. \ref{fig:bifurcation_num_alpha_m1}(a) ($\alpha = -1$) and \ref{fig:bifurcation_num_alpha_p1}(a) ($\alpha = +1$), bifurcation points are plotted with thick black curves. For both choices of $\alpha$, many families of bifurcations are identified. This a consequence of the nonlinear dependence on the parameter $\tilde{k}$ in the Jacobian matrix \eqref{eq:jacobian_new_vars}. 

The bifurcation points computed via the Jacobian are used as initial guesses in a simple numerical continuation routine that computes nontrivial solutions of the jump conditions \eqref{eq:jump_conds_newvars}. Three of the five far-field parameters in the jump conditions can be computed as a function of the remaining two. Using a parameterization similar to that found in  \cite{sprenger_discontinuous_2020}, we parameterize nontrivial solutions of the jump conditions by the left average $\ubar_-$  at $-\infty$, and the wave amplitude $a_+$ at $+\infty$. The results of the computation are two-dimensional  surfaces in $\R^3$  
\begin{align}
a_- = a_-(\ubar_-,a_+),\qquad k_-= k_-(\ubar_-,a_+), \qquad k_+ = k_+(\ubar_-,a_+),
\end{align}
a different surface bifurcating from each curve of bifurcation points, where $\ubar_-=0$ and $k_-=k_+.$ 
Examples of nontrivial solutions are displayed for fixed values of $a_+$ in Figs. \ref{fig:bifurcation_num_alpha_m1}(b)-(d) and \ref{fig:bifurcation_num_alpha_p1}(b)-(d). The upper panels of Figs. \ref{fig:bifurcation_num_alpha_m1}(b)-(d) and \ref{fig:bifurcation_num_alpha_p1}(b)-(d) show the left far-field amplitude, $a_-$ as a function of $\ubar_-$ and the lower panels show $k_-$ and $k_+$  as functions of $\ubar_-$ giving blue and red curves respectively. The continuation is terminated at values of $\ubar_-$ for which either $k_-$ or $k_+$ attain values less than $0.005$, which indicates one far-field periodic orbit is approaching a solitary wave; this case is discussed in detail in the following subsection. 

Figure \ref{fig:bifurcation_num_alpha_0}(a) shows the bifurcation points for the jump conditions of the KdV5 equation. The rightmost three curves of bifurcation points and the branches of nontrivial solutions were computed previously in \cite{sprenger_discontinuous_2020}, while the bifurcations curves that include the points labeled (b), (c) and (d) as well as the bifurcation curves to the left of these points were not previously identified. In Fig. \ref{fig:bifurcation_num_alpha_0}(b)-(d), example bifurcation curves are shown with $a_+ = 1$. 

\begin{figure}[h!]
\begin{center}
\includegraphics[scale=0.15,page = 2]{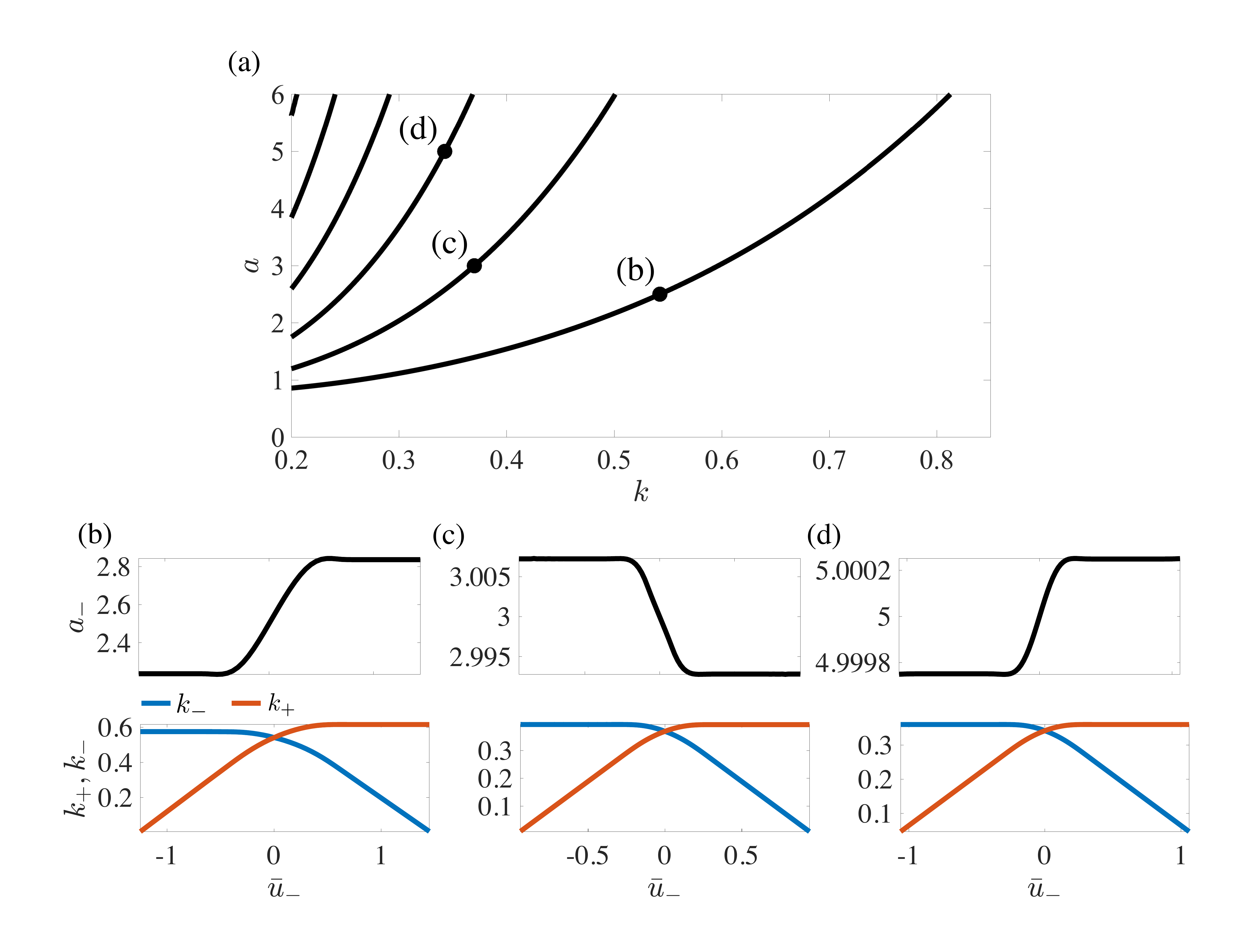}
\caption{($\alpha = -1$) (a) Bifurcation points of the jump conditions \eqref{eq:jump_conds_newvars}. (b)-(d) Wave parameters along the loci of nontrivial   solutions of the jump conditions, branching from bifurcation points (b-d) with fixed $a_+$. }
\label{fig:bifurcation_num_alpha_m1} 
\end{center}
\end{figure}

\begin{figure}[h!]
\begin{center}
\includegraphics[scale=0.15,page = 1]{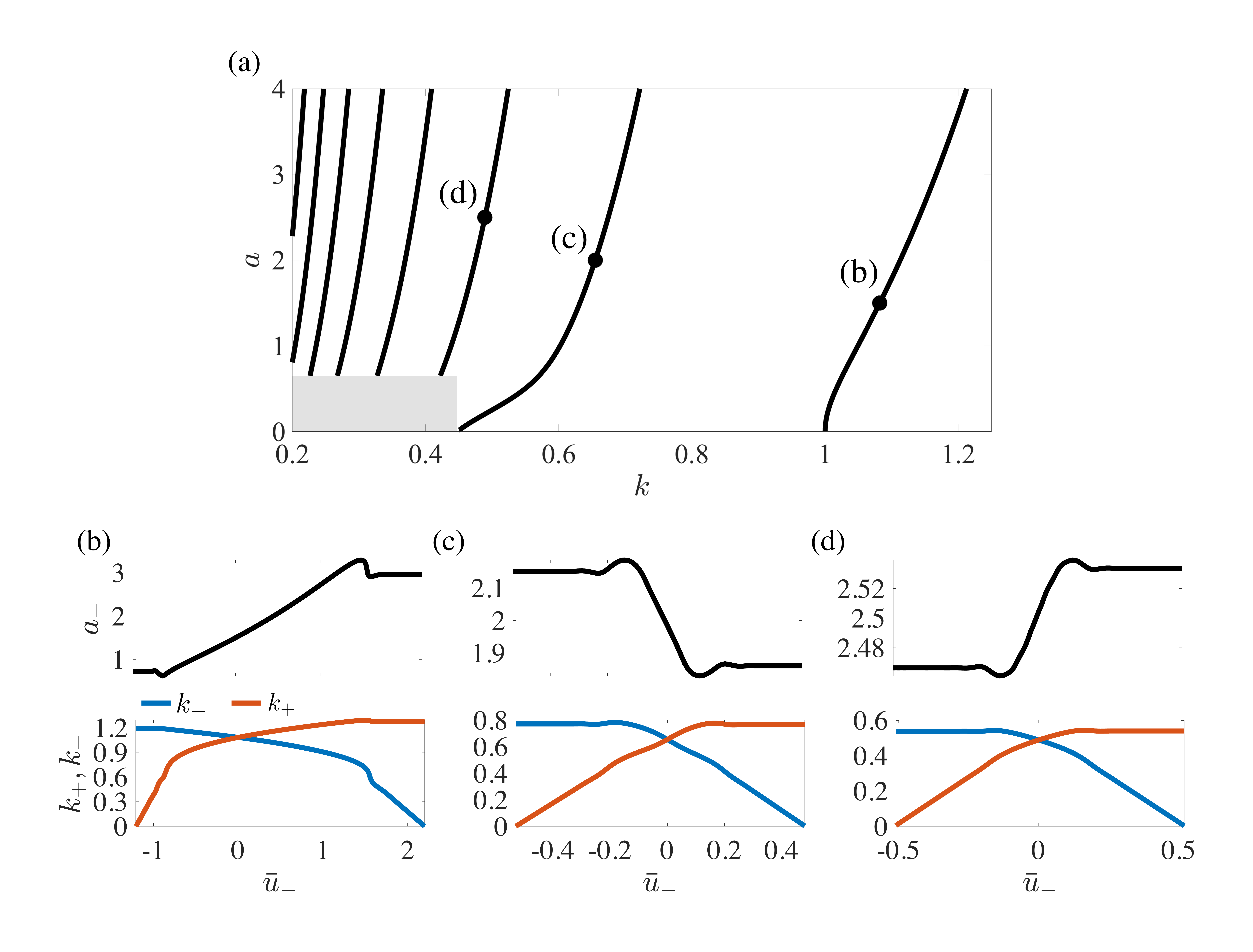}
\caption{($\alpha = +1$) (a) Bifurcation points of the jump conditions \eqref{eq:jump_conds_newvars}. (b)-(d) Wave parameters along the loci of nontrivial   solutions of the jump conditions, branching from bifurcation points (b-d) with fixed $a_+$. }
\label{fig:bifurcation_num_alpha_p1}
\end{center}
\end{figure}

\begin{figure}[h!]
\begin{center}
\includegraphics[scale=0.15,page = 3]{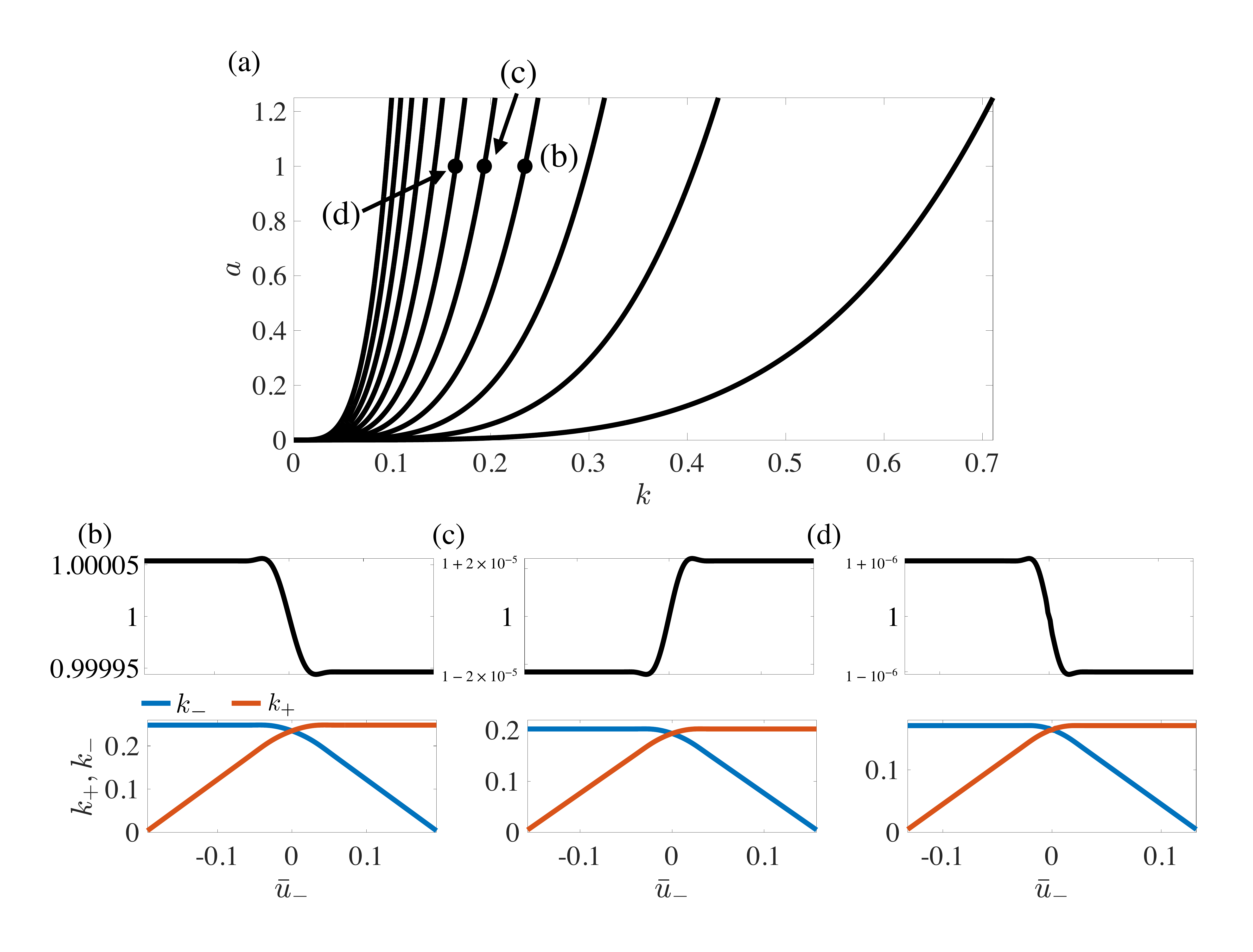}
\caption{($\alpha = 0$) (a) Bifurcation points of the jump conditions \eqref{eq:jump_conds_newvars}. (b)-(d) Wave parameters along the loci of nontrivial   solutions of the jump conditions, branching from bifurcation points (b-d) with fixed $a_+$. }
\label{fig:bifurcation_num_alpha_0}
\end{center}
\end{figure}

\subsection{Equilibrium to periodic wave} 
\label{sec:eq_to_p}
The curves of solutions of the jump conditions \eqref{jump_conditions}, shown in Figures \ref{fig:bifurcation_num_alpha_m1}--\ref{fig:bifurcation_num_alpha_p1} have the property that there are finite  limits for the parameter   $\ubar_-$, where  one of  the wavenumbers $k_+$ or $k_-$ approaches zero. In these  limits, the jump conditions \eqref{jump_conditions} are simplified. To be specific, we take $k_- \to 0$, so the periodic solution $f_-(\xi)=\varphi_-(k_-\xi;\ubar_-,a_-,k_-)$ approaches  a  homoclinic orbit  in the phase portrait, and corresponds to a solitary wave solution of the Kawahara equation \eqref{eq:kawahara}. Suppose $k_-\to 0$ as $\ubar_-\to\ubar_s.$ Let $c_s$ denote the speed of a solitary wave with amplitude $a_s$ defined by 
\begin{equation}\label{cs-eqn0}
c_s=c_s(\ubar_s,a_s)=\lim_{\ubar_-\to\ubar_s}c(\ubar_-,a_-(\ubar_-),k_-(\ubar_-))\,.\end{equation}
  Then $c_s(\ubar_s,a_s)=\ubar_s+ c_s(0,a_s),$ where  $c_s(0,a_s)$ is the speed of a solitary wave on zero background, with amplitude $a_s,$ which can be calculated numerically, independently of the limit \eqref{cs-eqn0}.

As $k_- \to 0$, the limiting homoclinic orbit of $\varphi_-$ approaches  $\ubar_s$ along the stable and unstable manifolds. Consequently, as $k_-\to 0,$  the $2\pi$-periodic $\varphi_-(\theta)$ is close to the constant $\ubar_s$   for much of the interval $0\leq \theta \leq 2\pi$ and all derivatives of   $\varphi_-$  are uniformly bounded. Therefore,  in the limit $k_-\to 0,$ the averages become,
\begin{align}\label{eq:zero_k_lim}
\overline{\varphi_-^n} = \ubar_s^n, \qquad 
\overline{k_-^{nm} \left(\frac{\partial^n \varphi_-}{\partial \theta^n}\right)^m} = 0, \ \ \text{for all positive integers} \  n,m. 
\end{align}
Thus, the jump conditions \eqref{jump_conditions} in this limit become
\begin{subequations}\label{eq:jump_conds_eq_to_per}
\begin{align}
-c_s \ubar_s + \frac{1}{2} \ubar_s^2 - \frac{1}{2}\overline{\varphi_+^2} & = 0,  \label{eq:e_to_p1}\\ 
-\frac{c_s}{2} \left(\ubar_s^2 -\overline{ \varphi_+^2}\right) + \frac{1}{3}\ubar_s^3 -\frac{1}{3}\overline{\varphi_+^3} + \frac{3}{2}\alpha k_+^2 \overline{\varphi_{+,\theta}^2} - \frac{5}{2}k_+^4 \overline{\varphi^2_{+,\theta\theta}}  & = 0 \label{eq:e_to_p2}  \\
c_+ - c_s & = 0,\label{eq:e_to_p3}
\end{align}
\end{subequations}
in which $c_+=c(0,a_+,k_+).$ Note that the only dependence on $a_s$ is in \eqref{eq:e_to_p3}, through $c_s(\ubar_s,a_s)$. Thus, if we write $c_s=c(0,a_+,k_+)$
from that equation, the first two conditions are equations for the variables $(\ubar_s, a_+, k_+).$ The numerical strategy is to first solve these two equations for $a_+,k_+$ as functions of $\ubar_s,$ and then solve 
\begin{equation}\label{cs-eqn}
c_s=\ubar_s+c_s(0,a_s)=c(0,a_+(\ubar_s),k_+(\ubar_s))\,,
\end{equation}
 for $a_s$ as a function of $\ubar_s.$ Here, the amplitude-speed relation  $c_s=c_s(0,a_s)$ for solitary waves on a zero background is obtained numerically \cite{sprenger_shock_2017}. 

As in the previous section, this procedure is made more explicit when the periodic orbit, $\varphi_+,$ is approximated by the Stokes expansion \eqref{eq:stokeswave}. In this regime, the first two jump conditions \eqref{eq:e_to_p1} and \eqref{eq:e_to_p2}, to $\mathcal{O}(a_+^2)$, are 
 \begin{subequations}\label{eq:eq_to_stokes}
\begin{align}
-\left(- \alpha k_+^2 + k_+^4 - \frac{a_+^2}{480 k_+^4-96 \alpha  k_+^2}\right)\ubar_s + \frac{1}{2}\ubar_s^2 - \frac{a_+^2}{16} &= 0 \\
-\left(- \alpha k_+^2 + k_+^4 - \frac{a_+^2}{480 k_+^4-96 \alpha  k_+^2}\right) \left(\frac{1}{2}\ubar_s^2 - \frac{a_+^2}{16}\right) + \frac{1}{3}\ubar_s^3 + \frac{3}{16}\alpha  a_+^2 k_+^2 - \frac{5}{16}a_+^2 k_+^4  &= 0.
\end{align}
\end{subequations}
These equations can be solved comprehensively for $a_+$ and $k_+$ as functions of $\ubar_s.$ Then $a_s(\ubar_s)$ is calculated using (\ref{cs-eqn}). However, the results yield only limited information about the full jump conditions \eqref{eq:jump_conds_eq_to_per}.

Returning to the full jump conditions \eqref{eq:jump_conds_eq_to_per} for equilibrium-to-periodic TWs, we find numerous curves of solutions $(\ubar_s,a_s, k_+)$ for fixed $a_+,$ 
computed using numerical averages of the periodic orbits, as in the previous section.  We fix a value of $\ubar_s$ and solve for the remaining   wave parameters $a_+$, $a_s$ and $k_+,$ using the results from the periodic-to-periodic section \S\ref{sec:numerical_shock} in the limit as $k_-\to 0$  to extract  starting points for an iterative  solver.  Then a path following procedure that yields $(a_\pm,k_+)$ as a function of $\ubar_s.$  A corresponding procedure can be used to calculate $(a_\pm,k_-)$ as a function of  $\ubar_-$ with $k_+(\ubar_-)\to 0$ as $\ubar_-$ approaches the other limit of its range.  These computations give the wave parameters at the boundaries  of the solution manifolds satisfying the jump conditions \eqref{jump_conditions},   identified by either $k_+$ or $k_-$ tending to zero. These boundaries in  $\ubar_-$  are observed in panels (b), (c) and (d) of Figures \ref{fig:bifurcation_num_alpha_m1}--\ref{fig:bifurcation_num_alpha_0} for solution curves bifurcating from corresponding points labeled (b), (c), (d) on the bifurcation curves, on which $\ubar_-=0, k_+=k_-=k, a_+=a_-=a.$ In the panels (b), (c), (d), $a_+$ is kept fixed, $u_-$ is varied between limiting values. As $\ubar_-\to \ubar_s, k_-\to 0,$ we can read off values of $k_+(\ubar_s)$ and $a_s=a_-(\ubar_s).$ These values depend also on $a_+.$ 
For each bifurcation curve in panel (a) of the figures \ref{fig:bifurcation_num_alpha_m1}--\ref{fig:bifurcation_num_alpha_p1}, there is a 2-dimensional surface of solutions of \eqref{jump_conditions}, parameterized by $\ubar_-, a_+$.  The edges of each such surface, where either $k_-\to 0$ or $k_+\to 0$ mark the curves where there is a solitary wave on the left or on the right. In the illustrative case $k_-\to 0,$ we obtain the curves $(a_s,a_+,k_+)(\ubar_s).$ Examples are plotted in Fig.\ref{fig:eq_to_p}.

Figure \ref{fig:eq_to_p}(a) ($\alpha = -1)$ shows the wave parameters $a_s,a_+,k_+$ as  functions of  $\ubar_s$. The curves represent the endpoints of the solutions that emerge from the rightmost set of bifurcation points in Figure \ref{fig:bifurcation_num_alpha_m1}(a). Panels (b) and (c) ($\alpha = +1$) are similar plots of the endpoint of the solutions emerging from the two rightmost   bifurcation points shown as (b), (c) in Figure \ref{fig:bifurcation_num_alpha_p1}(a). In panel (c), the inset for small values of $\ubar_s$ shows the solitary wave amplitude $a_s$ approaching sharply to zero.

\begin{figure}[h]
\begin{center}
\includegraphics[scale=0.66]{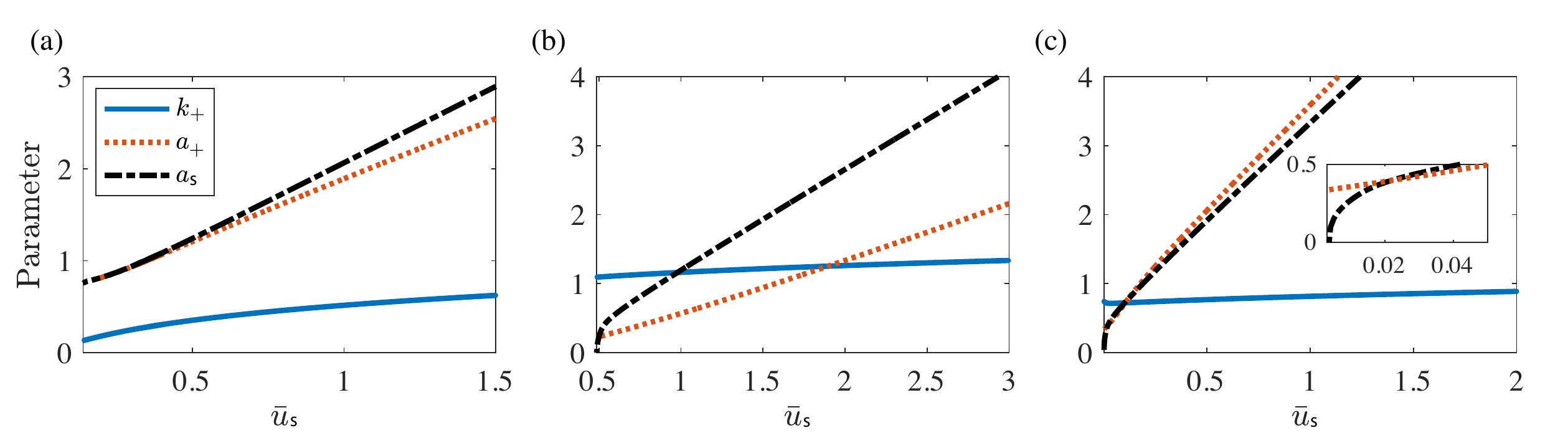}
\caption{Solutions of the jump conditions \eqref{eq:jump_conds_eq_to_per}.
(a) $\alpha = -1,$ from Fig.~\ref{fig:bifurcation_num_alpha_m1}(a).  (b),(c) $\alpha = +1,$ from Fig.~\ref{fig:bifurcation_num_alpha_p1}(a). }
\label{fig:eq_to_p}
\end{center}
\end{figure}

\subsubsection{Properties of solitary wave solutions}
\label{subsubsec-sws}

In order to ensure that solutions of the jump conditions \eqref{eq:jump_conds_eq_to_per} are necessary conditions for the existence of a traveling wave solutions of the Kawahara equation joining a solitary wave to a periodic solution, we identify parameter ranges in which we are guaranteed a solitary wave solution on the background $\ubar_s.$
For this we appeal to results on the existence
of homoclinic orbits in \cite{buffoni_bifurcation_1996}.

In \cite{buffoni_bifurcation_1996}, the authors examine homoclinic orbits of the equation
\beq\label{buffoni1}
u-u^2+Pu''+u^{(4)}=0.
\eeq
This equation has a unique homoclinic orbit for each $P\leq -2,$ and the authors compute multiple homoclinic orbits numerically for $-2<P<2.$     For $P<2,$ each orbit in $\R^4$ is the intersection of 2-dimensional stable and unstable manifolds     of the equilibrium $u=0.$ Linearizing the  equation \eqref{buffoni1} about the  equilibrium $\ubar=0,$ we find the characteristic equation $\lambda^4+P\lambda^2+1=0$ has four distinct real solutions $\pm\lambda_j, j=1,2,$ with    $0<\lambda_1<\lambda_2$ when $P<-2.$ For $-2<P<2,$ the solutions are complex, 
$\pm(\lambda \pm i\mu),$ and for $P>2,$ all four solutions $\pm i\mu_j, j=1,2$ are imaginary and distinct.  Equation \eqref{buffoni1} is equivalent to equation \eqref{eq:4th_ord_ode} with $A=0,$ corresponding to $\ubar_s=0$ background:
\beq\label{ode4}
-cf+\frac 12 f^2+\alpha f''+f^{(4)}=0.
\eeq
The equivalence follows by writing $\alpha =P b^{-2}, c=-b^{-4}, f=-\frac 12b^4u,$ where,  $b $ is the scaling between the independent variables $t,\xi$ of the ODEs \eqref{buffoni1} and \eqref{ode4} respectively: $\xi= bt.$ The existence of orbits homoclinic to zero for equation \eqref{ode4} with $\alpha=0,\pm 1$ can be summarized as follows, with reference to the corresponding results of Buffoni et al. \cite{buffoni_bifurcation_1996} and Amick and Toland \cite{amick_homoclinic_1992} in relation to equation \eqref{buffoni1}.  
\begin{enumerate}
\item $\alpha=-1\ (P<0):$  There is a unique symmetric homoclinic orbit for $-\frac 14 \leq c < 0 \ (P\leq -2),$ and at least one symmetric homoclinic orbit for each $c<-\frac 14 \ (-2<P<0).$
\item $\alpha=1\ (0<P<2):$ There is at least one symmetric homoclinic orbit for each $c<-\frac 14.$
\item $\alpha=0\ (P=0, b> 0 \ \mbox{arbitrary}):$\  There is at least one symmetric homoclinic orbit for all $c<0.$
\end{enumerate}
The statement in case 1 ($\alpha=-1$) is proved  in  \cite{amick_homoclinic_1992}; the existence of multiple homoclinic solutions  is explored numerically in \cite{buffoni_bifurcation_1996} in all three cases with $-2<P<2,$ including bifurcations between symmetric and asymmetric  homoclinic orbits as $P$ is varied. In our case,   the solitary wave is the limit as $k_-\to 0$ of periodic solutions, followed along a branch of solutions of wave parameters satisfying the jump conditions. The solitary wave speed $c_s=c_s(\ubar_s,a_s)=\ubar_s+c_s(0,a_s)$ depends on the background constant $\ubar_s$ and amplitude $a_s.$  
 
We observe in numerical results only solitary waves for which $P\in (-2,2).$ For $\alpha=+1,$ we necessarily have $P\in (0,2),$ so that $c_s<-\frac 14+\ubar_s.$ For    $\alpha=0, $ we find $c_s<\ubar_s,$ but since $P=0,$ the eigenvalues of the equilibrium $f=\ubar_s$ are complex.   However, for $\alpha= -1,$ the wave speed $c_s$ of solitary wave solutions of \eqref{ode4} may be either side of $-\frac 14+\ubar_s.$ In this case, numerical results show quite clearly that if the solitary wave is connected to a periodic solution, thus satisfying the jump conditions  \eqref{eq:jump_conds_eq_to_per}, then $c_s$ is limited to $c_s<-\frac 14+\ubar_s.$   These observations point to the interesting conjecture that if there is a traveling wave joining a solitary wave with background $\ubar_s$ to a periodic wave $\varphi_+$ (so the jump conditions \eqref{eq:jump_conds_eq_to_per} are satisfied),  then the eigenvalues of the  equilibrium   $f=\ubar_s$ of \eqref{eq:4th_ord_ode}  are two complex conjugate pairs. 
 
In the absence of a rigorous proof of the conjecture in case $\alpha=-1,$ we provide numerical evidence as follows. First, a periodic solution  $f$ of equation \eqref{ode4} 
with mean  $\overline{f}=0$ is found numerically for each amplitude $a$ and wavenumber $k,$ with $c=\tilde{c}(a,k).$ The procedure is described in Appendix~A.    Suppose a traveling wave solution of \eqref{ode4}  connects a periodic wave to an equilibrium  $\ubar_s.$ Then the corresponding parameters $a,k,\ubar_s$ satisfy the jump conditions \eqref{eq:jump_conds_eq_to_per}. Note that there is no assumption that $\ubar_s$ is the background for a solitary wave. However, since  the constant of integration $A=0$ in  \eqref{ode4},  we have  $\ubar_s = 0,$ and consequently  the pair $(a,k)$  lies in the zero set of the Hamiltonian: $H(a,k) =0$. The traveling waves joining different periodic solutions with zero Hamiltonian are indicated by pairs of intersections of level curves of $\tilde{c}(a,k) $ with the zero set of $H(a,k).$ In particular,  
the equation $H(a,0+)=0$ gives values of $a$ for which there is a solitary wave with amplitude $a=a_s$ and background $\ubar_s=0.$  Furthermore, the entire $a$ axis is accessible since $H=0$ for any solitary wave on background $\ubar_s =0$. From Figure~\ref{fig:hamil_zero_constant}, we observe that a level curve of $c(a,k)$ may fail to intersect any of the curves representing the zero set of $H(a,k).$ Indeed, we have shown in the figure that the curve $c(a,k)=-\frac 14$ fails to intersect the curves $H(a,k)=0, $ except asymptotically at $k=0.$ Representative level curves with velocities $c = -0.5 $ and $c = -0.1$ are also shown in the figure. We conclude from these calculations that if $\ubar_s=0$ is connected by a traveling wave to a periodic solution, then the wave speed $c$ satisfies $c<-\frac 14.$ The general case is then established by the following proposition:

\begin{proposition}
Suppose that every traveling wave joining the equilibrium $\ubar_s=0$ to a periodic wave solution $f(\xi)$ of equation \eqref{ode4} has speed $c_0<-\frac14,$ as established with numerical results reported in Figure~\ref{fig:hamil_zero_constant}. Then traveling waves joining  any equilibrium $\ubar_s\in \R$  to a periodic wave solution necessarily have speed $c_s<-\frac14+\ubar_s.$  
\end{proposition}\

\begin{proof} Consider $\ubar_s\in
 \R,$ and suppose a traveling wave solution $f(\xi)$ with speed $c$ is asymptotic to $\ubar_s$ as $\xi\to -\infty$, and to a $2\pi/k$-periodic traveling wave $f_+(\xi)$ as $\xi\to +\infty.$ Then  $\ubar_s,k,c$  and $\varphi_+(\theta)=f_+(\theta/k)$ satisfy the jump conditions  \eqref{eq:jump_conds_eq_to_per}. Let 
 $$
 A=-c\ubar_s+\frac 12 \ubar_s^2.
 $$ 
 Then $\ubar_s=c+\sqrt{c^2+2A}$ (plus sign chosen so that $\ubar_s=0$ when $A=0$ and $c<-\frac14$), and $f_+(\xi)$ satisfies equation  \eqref{eq:4th_ord_ode} with constant of integration $A.$   Next, we use property (2) of Lemma~\ref{Lemma_1} to change the constant of integration in \eqref{eq:4th_ord_ode} to zero.   Specifically, let 
 $$
 \tilde{c}=-\sqrt{c^2+2A}, \   \ \tilde{f}=f-\ubar_s, \ \ \mbox{and} \ \ \tilde{f}_+=f_+-\ubar_s.
 $$ 
 Then $\tilde{f}, \tilde{f}_+$  satisfy \eqref{eq:4th_ord_ode}  with constant of integration $A=0$ and speed $\tilde{c}.$ Moreover, $\lim_{\xi\to -\infty} \tilde{f}(\xi)=0.$
 Thus, by hypothesis,
 $\tilde{c}<-\frac14.$ But  $\ubar_s=c-\tilde{c},$ so that $c<-\frac 14 +\ubar_s,$ as claimed. 
 \end{proof}

The level curves $H(a,k)=0$ shown in Figure~\ref{fig:hamil_zero_constant} are calculated numerically. In the figure, the numerical value of the Hamiltonian is below double machine-precision when the wavelength of the periodic orbit is sufficiently large, so the contours cannot be reliably computed for small values of $k$. To overcome this difficulty, we extrapolate contour lines for small values of $k.$ To extrapolate the contours, we assume that at $k = 0$, the curves meet at the amplitude of the solitary wave with velocity $c = -1/4$. Each contour is a curve $a = a(k)$ such that $H(a(k),k)=0.$ We observe that  $a(k)$ is exponential in $k$ sufficiently far away from $k = 0.$ To extrapolate from the computed portion of the graph of $a(k),$ we fit a cubic spline to the data, with an additional data point $k=0, \ a \approx 0.72718.$
This value of $a$ is the approximate amplitude of a numerically computed solitary wave with  velocity $c = -1/4.$ We also ensure that $a'(0)=0$ 
by using an even reflection of the data across the $a$-axis. In Figure \ref{fig:hamil_zero_constant}, we represent the numerically computed level set as solid black curves, and the extrapolated data as dashed curves. The inset shows the same curves near $k = 0$ with the horizontal axis on a logarithmic scale. 
 
 \begin{figure}[h]
\begin{center}
 \includegraphics[scale=0.4]{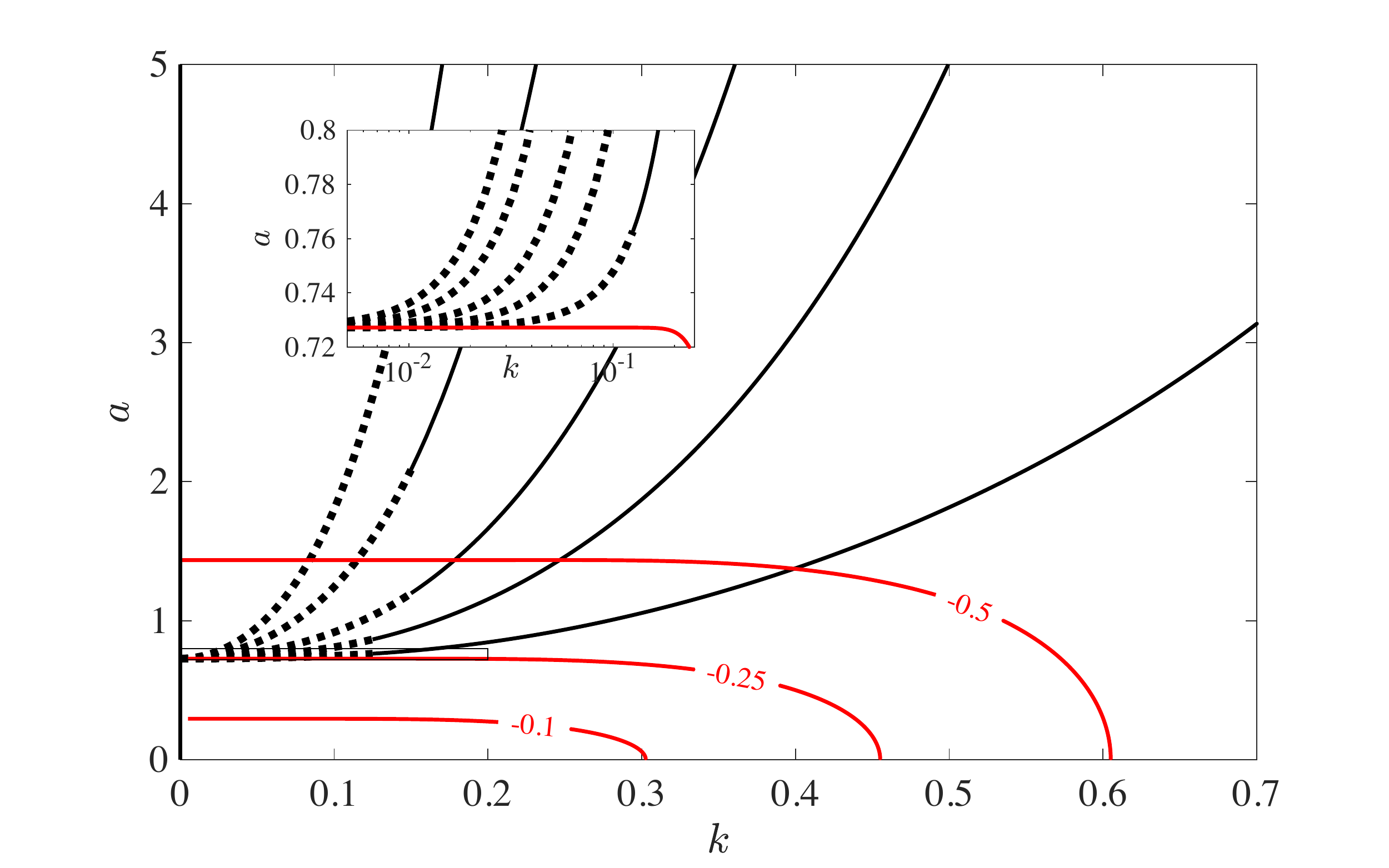}
 \caption{Level curves $H = 0$ of the Hamiltonian for periodic solutions of Eq. \eqref{eq:4th_ord_ode} with $A = 0.$ The   solid black curves are computed,  and the dashed curves extrapolated to  small values of $k.$   The red level curves $\tilde{c}(a,k)=c$  correspond to  periodic solutions with $c =-0.1, -0.25, -0.5$ and $A=0.$ The inset expands the thin rectangular area around $a=0.73,$ with $k\sim 0$ shown on a log scale. }
 \label{fig:hamil_zero_constant}
\end{center}
 \end{figure} 

\section{Computations of heteroclinic connections between periodic orbits}
\label{sec:heteroclinic_computations}

To complete the construction of traveling waves, we identify pairs $\varphi_\pm$ of periodic orbits satisfying the jump conditions of the previous section, and compute a solution $f(\xi)$ of the traveling wave ODE \eqref{eq:profile_eq} as the intersection of stable and unstable manifolds of the periodic orbits. This construction works smoothly providing both orbits are {\em hyperbolic}, meaning that they each have four real Floquet multipliers.  If one of the orbits has two non-real multipliers, as in Figure \ref{fig:floquet_mult_config}(c), we would need to calculate orbits on an invariant torus, leading to a much more complex picture of traveling waves. According to Figure \ref{fig:alpha_m1_floquet}, we should choose values of $a_\pm, k_\pm$ on branches of solutions away from the bifurcation points, avoiding shaded regions in the figure. Comprehensive details of the procedure to compute invariant manifolds of periodic orbits are provided in \cite{koon_dynamical_2011}, where a similar construction is used in the context of a restricted three-body problem. 

Consider the $2\pi$-periodic solution $\varphi_-(\theta)$ of the ODE \eqref{ode_1}, and let $f_-(\xi)=\varphi_-(k\xi).$ Recall from  \S\ref{floquet1} that the  ODE linearized about $f_-$ defines a flow map $\Phi(\xi), 0\leq \xi<2\pi/k.$ Since $\varphi_-$ is assumed to be hyperbolic, $f_-$ has a real Floquet multiplier $\lambda$ with $|\lambda|>1$   and corresponding eigenvector $v_\lambda\in\R^4$ of   the monodromy matrix $M=\Phi(2\pi/k)$, so that $Mv_\lambda=\lambda v_\lambda.$  The two-dimensional unstable manifold $W^u$ of $f_-$ is computed by solving the ODE \eqref{eq:4th_ord_ode} written as a 1st-order system with $f_1=f(\xi):$

\begin{align}\label{eq:4th_ord_system}
\frac{d}{d\xi}\mathbf{f}= 
\begin{bmatrix}
f_1 \\ f_2 \\ f_3 \\ f_4 
\end{bmatrix} ' = \begin{bmatrix}
f_2 \\ f_3 \\ f_4 \\ c f_1 - \frac{1}{2}f_1^2 - \alpha f_3 + A
\end{bmatrix},
\end{align}
with  initial conditions depending on a parameter $\delta,$ with $0 < \delta \ll 1$ 
\begin{align*}
\mathbf{f}(0)=\mathbf{f_-}(0)\pm \delta \mathbf{v}_{\lambda}.
\end{align*}
Here,   $\mathbf{f_-}(0)=(\varphi, k\dot{\varphi},k^2\ddot{\varphi},k^3\dddot{\varphi})(0), \ \dot{ } = d/d\theta.$   This generates trajectories $\mathbf{f}^\pm_0(\xi), \xi\geq 0$ and, by also varying $\delta>0,$ the two dimensional unstable manifold $W^u$.  

Instead of parameterizing the invariant manifold by $\xi$ and $\delta,$ we can instead parameterize by $\xi $ and $\theta.$ This parameterization can be implemented by shifting the initial point on the periodic orbit $\mathbf{f_-}(\xi).$ Let $\theta\in [0,2\pi),$ and consider the initial point $\mathbf{f_-}(\xi_0),$ with $\xi_0=\theta/k.$ Then $M\Phi(\xi_0)v_\lambda=\lambda \Phi(\xi_0)v_\lambda,$ so we solve \eqref{eq:4th_ord_system} with initial condition 
$$
\mathbf{f}(\xi_0)=\mathbf{f_-}(\xi_0)\pm \delta \Phi(\xi_0)\mathbf{v}_{\lambda},
$$
thereby generating a further pair of trajectories $\mathbf{f}^\pm(\xi;\theta), \xi\geq \xi_0.$ In this way, we have generated the two-dimensional unstable manifold $W^u$ of $\mathbf{f_-},$   parameterized by  $\xi, \theta.$ The topology of the surface depends on the sign of the nontrivial multiplier with largest magnitude. If $\lambda > 1$, the orientable unstable manifold is topologically a cylinder.  A computed  example for this case is shown in Figure \ref{fig:manifolds_3d}(a). For $\lambda < -1$, $W^u$ is a M\"obius band, as explained in \cite{aougab_isolas_2019}, and  is therefore nonorientable. A computed example for this case is shown in Figure \ref{fig:manifolds_3d}(b). We will later in this section compute the heteroclinic connection between these two periodic orbits. 

 \begin{figure}[h]
\begin{center}
\includegraphics[scale=0.6]{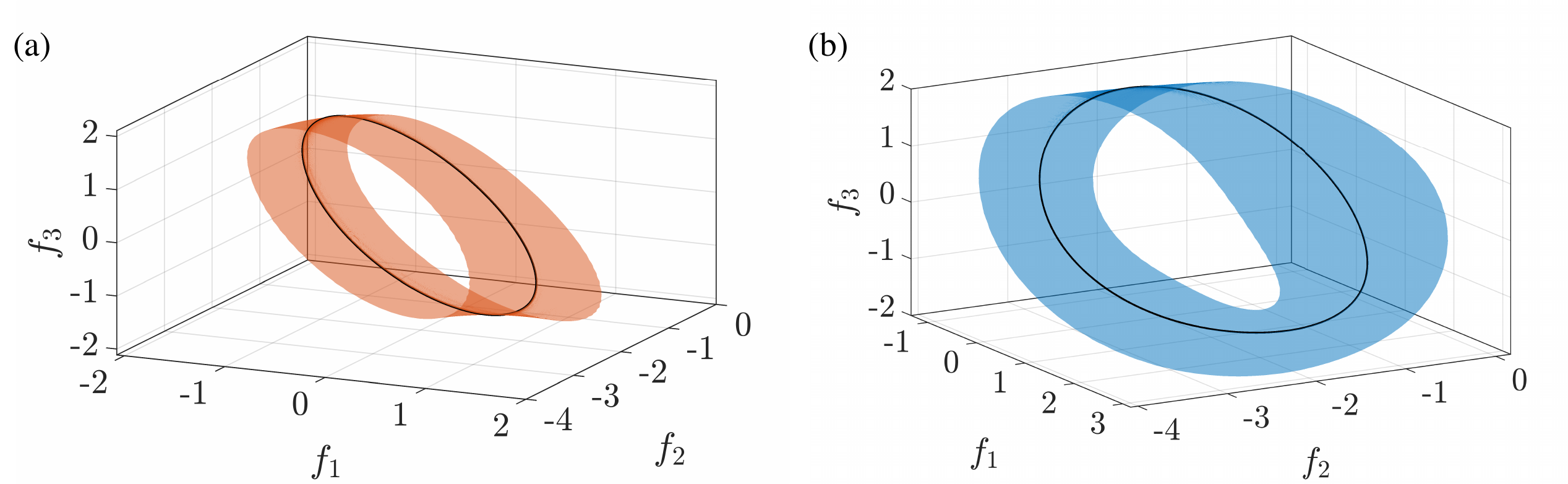}
\caption{($\alpha=+1$) Unstable manifolds for periodic orbits. (a) Orientable stable unstable manifold, $W^{s}$, of the periodic orbit $f_-$ with parameters $(\ubar,a,k) \approx  (0,2,1.2582)$, (b) Nonorientable unstable manifold, $W^{u}$, of the periodic orbit $f_+$ with parameters $(\ubar,a,k) \approx (1,3.2713,0.9736)$. The periodic orbits are the solid black curves.  
}
\label{fig:manifolds_3d}
\end{center}
\end{figure}

The procedure is similar to compute the stable manifold $W^s$ of $\mathbf{f_+},$ for which the initial perturbation is taken in the direction of $\mathbf{v}_{1/\lambda}$. The two dimensional invariant manifolds $W^u, W^s$ reside on the 3-dimensional hypersurface  $H(f_1,f_2,f_3,f_4)= constant,$ defined in \eqref{eq:energy_integral}.

In the 4-dimensional phase space, the level surface of the Hamiltonian \eqref{hamiltonian1} is a 3-d manifold. In this 3-d parameter space, the stable and unstable manifolds  $W^s$ and $W^u$ are two-dimensional, so they generically intersect in a curve. If we consider both manifolds to be parameterized through $(\xi,\theta),$ as described  above,  then different points on the curve  correspond to translations of $\theta,$ and corresponding translations of $\xi.$ 
 
The next step is to define a suitable Poincar\'{e} section ${\cal{P}}: f=constant.$ In figures \ref{example_TW_construction}--\ref{fig:eq_to_per_computations}, we use  $f = \frac{1}{2}\ubar_-$, where $\ubar_-$ is the average of the periodic orbit $\varphi_-.$ The manifold $W^u$  intersects $\cal{P}$ in a curve, and for $\varphi_+$ close enough to $\ubar_-,$ the stable manifold of $\mathbf{f_+}$ will also. Since both manifolds lie in the same 3-dimensional level surface $\cal{H}$ of $H,$ these curves lie in the same two dimensional surface, namely $\cal{H}\cap\cal{P}.$ Providing the curves intersect, the point of intersection represents a trajectory joining $\mathbf{f_-}$ to $\mathbf{f_+}.$  

To reconstruct the connecting orbit, we appeal to the autonomous property of the ODE system to set $\xi=0$ at the intersection point, then trace the trajectory on $W^u$ through negative $\xi$ back to $\mathbf{f_-},$ and the trajectory through positive $\xi$ to $\mathbf{f_+}.$ Since these two trajectories are not exact solutions, the constructed orbit has to be truncated by limiting $\xi$ to a bounded interval.

To compute example traveling wave solutions, we use wave parameters that satisfy the nonlinear jump conditions \eqref{eq:jump_conds_newvars}, thereby identifying the far-field periodic orbits $\varphi_-$ and $\varphi_+$. The corresponding nontrivial Floquet multipliers are calculated and we label those with largest modulus as $\lambda_-$ and $\lambda_+$ respectively. The  far-field wave parameters used in our example computations are summarized in Table \ref{tab:comput_params}, and the traveling wave solutions are shown in Figures \ref{example_TW_construction}--\ref{fig:eq_to_per_computations}. The parameters in the table are chosen to cover the two cases $\alpha=\pm1$ of the Kawahara equation, and for each case,  two examples of traveling waves are shown, one in which the Floquet multipliers of both periodic orbits are positive, and the other in which one orbit has negative nontrivial multipliers.

\begin{table}[h]
{\normalsize
\begin{center}
\begin{tabular}{c| c c c c c c}

Figure & $\alpha$ & $(\ubar_-,a_-,k_-)$ & $(\ubar_+,a_+,k_+)$ & $\lambda_-$ & $\lambda_+$ \\ \hline \\ [-6pt] 
\ref{example_TW_construction} & $+1$ &$(1,3.2713,0.9736)$ & $(0,2,1.2582)$ & $-74.4664$ & $38.7253$ \\ 
\ref{fig:example_TWs}(a) & $+1$ & $ (2,3.8003,0.5382)$ & $(0,2,1.3222)$ & $1.0399 \times 10^3$ & $56.2029$  \\ 
\ref{fig:example_TWs}(b) & $-1$ &  $(0.05,2.5234,0.5382)$ & $(0,2.5,0.5484)$ & $-2.2494 \times 10^4$ & $1.5485 \times 10^4$\\ 
\ref{fig:example_TWs}(c) & $-1$ &  $(0.25,2.1710,0.4057)$ & $(0,2,0.5279)$ & $-7.1178 \times 10^4$ & $1.7536 \times 10^4$\\ 
\ref{fig:eq_to_per_computations}(a)-(c) & $+1$ & $ (2.2041,2.9627,0)$ & $(0,2,1.3236)$ & --- & $56.612$  \\
\ref{fig:eq_to_per_computations}(d) & $-1$ & $ (1.4646,2.8366,0)$ & $(0,2.5,0.6162)$ & --- & $5.6118 \times 10^4$  
\end{tabular}
\end{center}
\caption{Parameters of far-field periodic orbits connected by a heteroclinic orbit. }
\label{tab:comput_params}
}
\end{table}

In Figure \ref{example_TW_construction}, we plot portions of the invariant manifolds for the traveling wave limiting to the periodic wave $\mathbf{f}_-$ with the parameter triple $(\ubar_-,a_-,k_-) \approx (1,3.2713,0.9736)$ as $\xi \to -\infty$ and limiting to the periodic wave $\mathbf{f}_+$ with parameter triple $(\ubar_+,a_+,k_+) \approx (0,2,1.2582)$ as $\xi \to + \infty$. These wave parameters result in the invariant manifolds shown in Fig. \ref{fig:manifolds_3d}. In panel \ref{example_TW_construction}(a), we show the relevant portions of the invariant manifolds with the transparent red and blue surfaces. The colors match those in Fig. \ref{fig:manifolds_3d}, so the unstable manifold of $\mathbf{f}_-$ is shown in blue and the stable manifold of $\mathbf{f}_+$ is shown in red. The corresponding periodic orbits $\mathbf{f}_-$ and $\mathbf{f}_+$ are the dashed blue and red curves respectively. The trajectories along the invariant manifolds are integrated until they intersect transversally on the Poincare section $f_1 = \frac{1}{2}\ubar_-$. Panel \ref{example_TW_construction}(b) is a zoom in near the transverse intersection in the $(f_1,f_2,f_3)$ phase plane. The 3-dimensional figure in panel \ref{example_TW_construction}(c) shows the intersection of the invariant manifolds with the Poincar\'{e} section $f_1= \frac{1}{2}\ubar_-,$ together with the level surface of the Hamiltonian \eqref{hamiltonian1} $H = H_0$ (computed from $f_-$), evaluated at $f_1 = \frac{1}{2}\ubar_-$. This two dimensional surface is given by the equation
$$ 
\frac{\alpha  }{2} f_{2}^2 +f_{2}f_{4}-\frac{1}{2}f_{3}^2=H_0 +\frac{c}{2}\left(\frac{\ubar_-}{2}\right)^2 - \frac{1}{6}\left(\frac{\ubar_-}{2}\right)^3 + A\frac{\ubar_-}{2},
$$
 a conic section shown as a translucent gray surface. The unstable manifold of $\mathbf{f}_-$ is shown as  a blue curve and the stable manifold of $\mathbf{f}_+$ is the red curve, matching the colors in panel \ref{example_TW_construction}(a). From panels \ref{example_TW_construction}(b) and \ref{example_TW_construction}(c), we determine the trajectories along the invariant manifolds that intersect and the reconstructed traveling wave solution is plotted in panel \ref{example_TW_construction}(d). In Fig.~\ref{fig:example_TWs} we plot  traveling wave solutions computed via the same process as for Fig.~\ref{example_TW_construction}, but with parameters (listed in Table \ref{tab:comput_params}), representing different cases.

\begin{figure}[h]
\begin{center}
\includegraphics[scale=0.4]{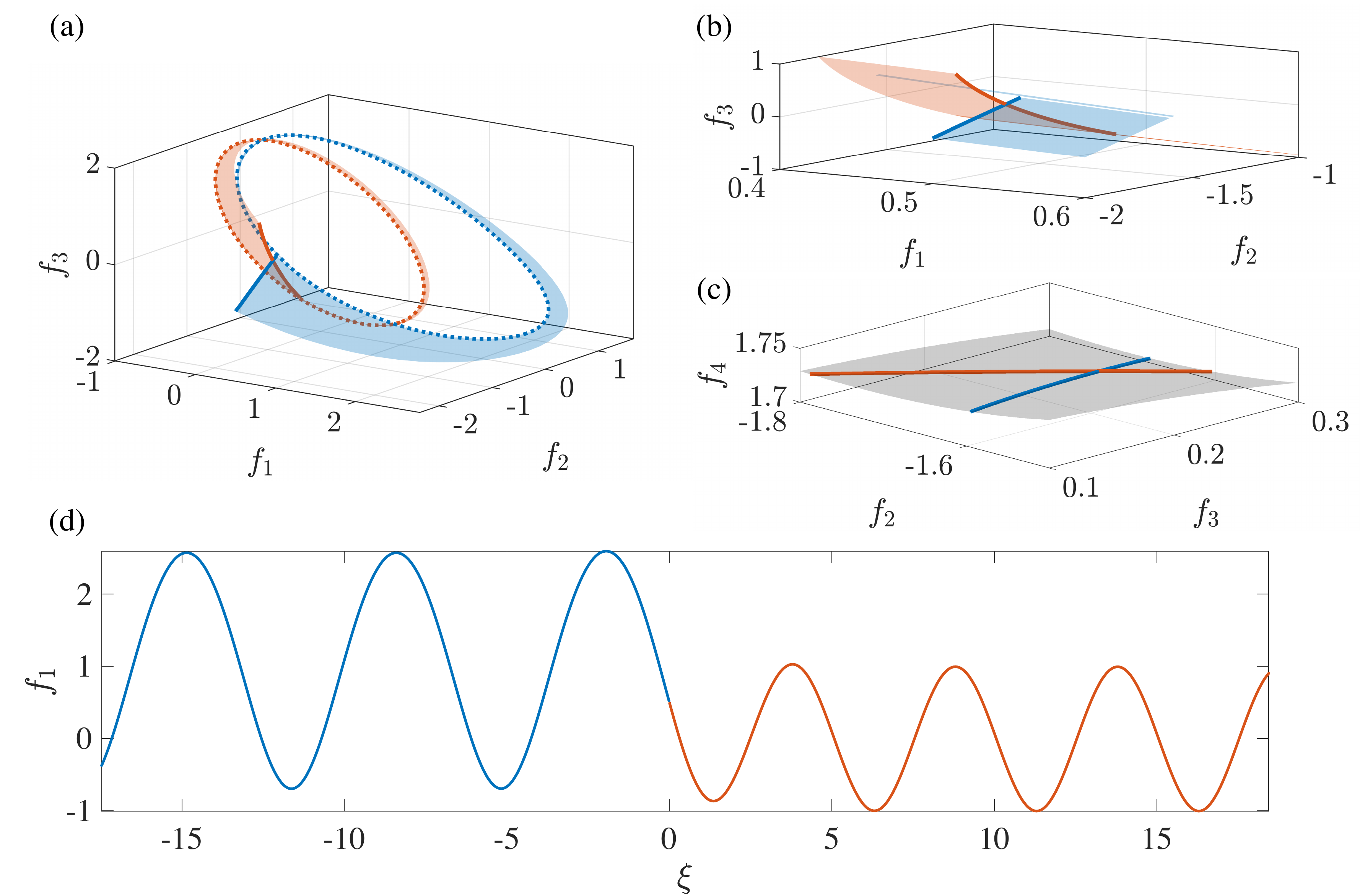}
\caption{Example traveling wave solution of the Kawahara equation \eqref{eq:kawahara} constructed from the intersection of invariant manifolds with $\alpha +1$. (a) Numerical computation of the intersecting invariant manifolds (b) Zoom-in of the invariant manifolds near their intersection on the Poincar\'{e} section $f_1 = \frac{1}{2}\ubar_-$ (c) Transverse intersection of the invariant manifolds (red and blue curves) on the level set of the Hamiltonian (d) Reconstructed traveling wave solution.} 
\label{example_TW_construction}
\end{center}
\end{figure}

\begin{figure}[h]
\begin{center}
\includegraphics[scale=0.4]{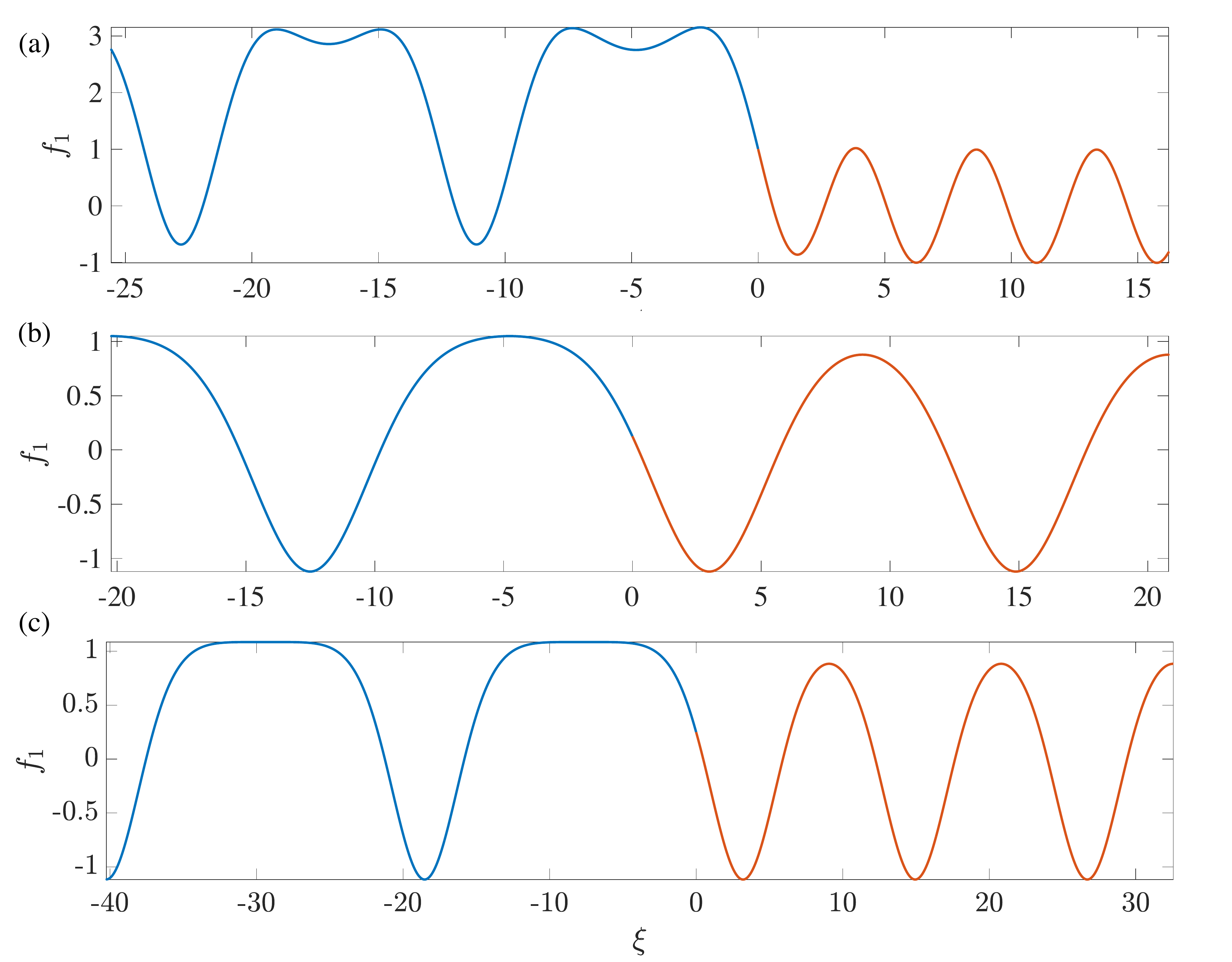}
\caption{Example traveling wave solutions. See Table \ref{tab:comput_params} for the far-field periodic wave parameters. } 
\label{fig:example_TWs}
\end{center}
\end{figure}
   
Examples of traveling waves that connect a solitary wave $f_-$ with constant background  to a periodic wave $f_+$ are shown in Figure \ref{fig:eq_to_per_computations}. In those cases, the traveling wave orbit follows the unstable manifold of the solitary wave with background $f_1^-$ before connecting to the orbit of $\varphi_+.$ The solitary wave is shown in panel~(c) superimposed on the graph of $f(\xi)$ as a dashed line.  The structure of the traveling wave is similar for $\alpha=+1$ and $\alpha=-1,$ but in Figure  \ref{fig:eq_to_per_computations}(d), oscillation within the solitary wave is clearly visible, whereas in Figure \ref{fig:eq_to_per_computations} there are further oscillations that cannot be seen at the scale shown. As mentioned previously in \S \ref{sec:eq_to_p} solutions joining equilibria to periodic waves can only be computed for sufficiently negative velocities, $c - \ubar_s < - 1/4$; the traveling wave shown in Fig. \ref{fig:eq_to_per_computations}(c) and (d) move with the respective velocities $c - \ubar_s \approx -1.49$ and  $c - \ubar_s \approx -0.99$.

\begin{figure}[h]
\begin{center}
\includegraphics[scale=0.4]{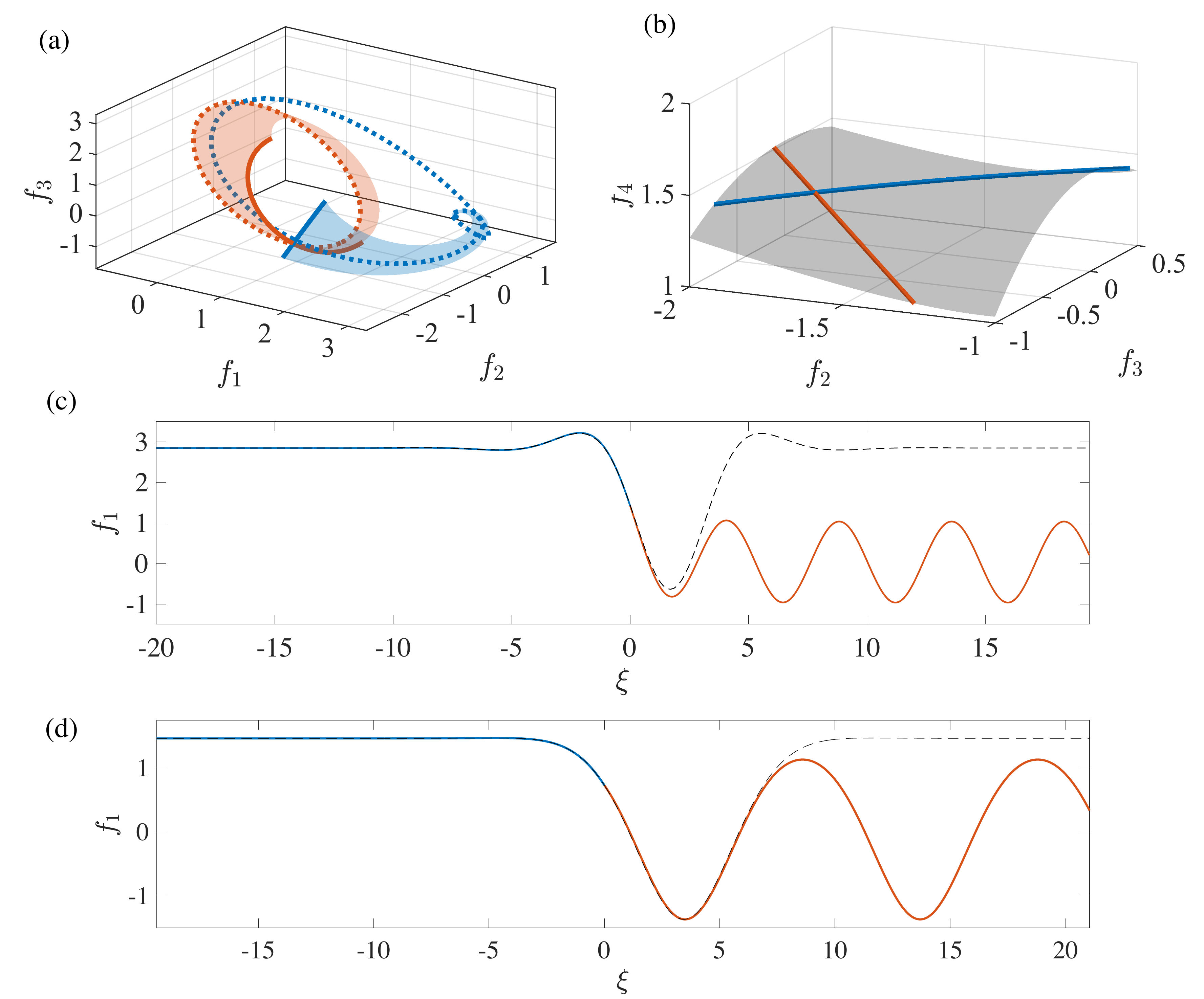}
\caption{Construction of traveling waves joining a constant to a periodic wave, through a solitary wave. Panels (a)-(c) correspond to the same traveling wave with $\alpha = +1$ (a) Numerical computation of the intersecting invariant manifolds (b) Transverse intersection of the invariant manifolds (red and blue curves) on the level set of the Hamiltonian with $f_1 = \frac{1}{2}\ubar_s$ (c) Reconstructed traveling wave solution. (d) representative equilibrium to periodic traveling wave for $\alpha = -1$. }
\label{fig:eq_to_per_computations}
\end{center}
\end{figure}

\section{Discussion and Conclusions}

In this paper, we explore the presence of traveling wave solutions of the Kawahara equation represented by heteroclinic orbits of the associated fifth order ordinary differential equation. Since the ODE can be integrated, the phase space for the orbits is four dimensional. However, the dimension is further reduced due to the Hamiltonian structure of the first order system, so that orbits can be represented using three-dimensional parameterizations of level surfaces of the Hamiltonian. An unusual aspect of the problem is that the traveling waves connect either two periodic waves or a solitary wave (i.e., a homoclinic orbit)  and a periodic wave. We use a combination of theoretical understanding of the structure of the phase portraits, and constructive numerical techniques, to explain the occurrence of the traveling waves. 

Interestingly, the initial goal of the research was to understand the role of  of shock wave solutions of Whitham modulation equations in the context of the fifth order KdV equation. In this special case of the Kawahara equation,  the recent work of Sprenger and Hoefer \cite{sprenger_discontinuous_2020} formulates jump conditions for the shock waves and explores the observation that they are related to the existence of traveling waves between periodic orbits. Here we approach the existence of traveling
waves between periodic orbits directly, without recourse to Whitham theory.
 
Numerical techniques rely on a pseudospectral method, representing the periodic solutions and the ODE by truncated Fourier series. Theoretical tools include the observation of scaling properties of the equations, dynamical systems properties of the Hamiltonian system such as the structure of stable and unstable manifolds of periodic solutions,  the use of a Poincar\'{e} section, formulation of compatibility conditions between the asymptotic states and their analysis with bifurcation  theory and parameter continuation. 

Our investigations depend on an assumption (item 1 below), and lead to a conjecture (item 2) 
\begin{enumerate}

\item The construction of traveling waves between periodic solutions depends on both periodic orbits being hyperbolic, meaning that all four of the Floquet multipliers (of each orbit) are real, and two are distinct. However, we show that if the periodic solutions are sufficiently close in their wave parameters, then one of the orbits has two complex conjugate Floquet multipliers on the unit circle. For this orbit, we are unable to characterize the stable or unstable manifolds, so that the dynamics of such traveling waves between the periodic solutions are not available without further understanding. This is a consequence of the fact that the bifurcation in the jump conditions occurs near a critical point of the Hamiltonian, as a function of the wavenumber.

\item In constructing traveling waves joining a constant (equilibrium) solution of the ODE system to a periodic solution, the compatibility conditions between the constant and the periodic wave are the same as for a solitary wave to be joined to the periodic solution, the solitary wave having the equilibrium as background. In this circumstance we find numerically that the speed of the traveling wave is limited to values for which the equilibrium has complex, non-real eigenvalues. We explain how this scenario is related to well-known properties of solitary waves for the Kawahara equation, and conjecture that the limitation is correct, specifically: if there is a traveling wave joining a constant to a periodic wave, then the wave speed is such that the eigenvalues of the constant are non-real. We reduce the conjecture to testing it on a one-parameter family of solutions, and provide numerical evidence for that family.
\end{enumerate}
 
\section*{Acknowledgements}
The research of PS and MS was supported by National Science Foundation grant  DMS-1812445. 

\bibliographystyle{siam}
\bibliography{references}

\renewcommand{\theequation}{A-\arabic{equation}}
  \setcounter{equation}{0}  
 \section*{Appendix A: Numerical Computation of Periodic Orbits}
In this appendix, we discuss the numerical approximation of periodic solutions via a pseudospectral method similar to that used in \cite{ehrnstrom_traveling_2009}.  We first set the constant of integration $A=0$ in the fourth order ODE \eqref{eq:4th_ord_ode}:
\begin{align}\label{eq:A1}
-c f + \frac{1}{2}f^2 + \alpha f'' + f'''' = 0.
\end{align}
For each wavenumber  $k,$    $2\pi/k$-periodic solutions $f$ of \eqref{eq:A1} are approximated by a truncated Fourier series 
\begin{align}
f \approx F_N = \sum_{n = -{N}}^{N} \hat{f}_n e^{ink\xi}. 
\end{align}
Substituting into \eqref{eq:A1} gives the nonlinear  equation
\begin{align}\label{eq:A3}
 \frac{1}{2}F_N^2 +  \sum_{n = -{N}}^{N} \left(-c + \alpha (nk)^2 + (nk)^4\right)\hat{f}_n e^{i n k \xi} = 0. 
\end{align}
The projection of  this equation onto each Fourier mode $e^{ink\xi}$, $n = -N,\ldots, N$ results in a system of $2N+1$ equations for the $2N+1$   Fourier coefficients $\hat{f}_n,$ with constant  $c,$ treated as a continuation parameter.
 
From  linear theory (linearizing \eqref{eq:A1} about $f=0$), we have infinitesimal $2\pi/k$-periodic solutions  with   phase velocity $c_{\rm p}(k) = - \alpha k^2 + k^4$. Approximate finite amplitude solutions are computed using  Matlab's nonlinear solver fsolve, choosing $N$  large enough, depending on $k,$  to push the residual below  $10^{-12}$. For example, for $k = 1$,  $2^6$ Fourier modes are required, while $2^{12}$ Fourier modes are needed for $k = 0.005$. The solutions $f=\tilde{f}$ are found by continuation from the small amplitude solutions as  $c$ varies away from $c=c_{\rm p}(k).$   This gives  the periodic wave amplitude $a=a(c, k)$ and average $\tilde{u}(c,k)$ as functions of its velocity and wavenumber. Inverting the relation $a=\tilde{a}(c,k))$ for each $k,$ and interpolating  using cubic splines, gives the velocity $c=\tilde{c}(a,k),$ and average $\ubar=\tilde{u}(a,k).$  

In a final step, we can use the Galilean symmetry \eqref{eq:galilean} to shift the mean of the periodic solutions to $\ubar =0,$ thereby modifying the wave velocity to $c(a,k)= \tilde{c}(a,k) - \ubar(a,k).$  In so doing, we obtain the  two-parameter family of periodic solutions used throughout this manuscript. 

\end{document}